\documentclass[11pt,twoside]{article}
\usepackage[utf8]{inputenc}
\usepackage[hmargin=1in,vmargin=1in]{geometry}

\usepackage{jeffe}
\usepackage{graphicx}
\DeclareGraphicsExtensions{.pdf,.png,.jpg}

\usepackage[charter]{mathdesign}
\usepackage{berasans,beramono}
\usepackage[mathcal]{eucal}
\SetMathAlphabet{\mathsf}{bold}{\encodingdefault}{\sfdefault}{b}{\updefault}
\SetMathAlphabet{\mathtt}{bold}{\encodingdefault}{\ttdefault}{b}{\updefault}
\SetMathAlphabet{\mathsf}{normal}{\encodingdefault}{\sfdefault}{\mddefault}{\updefault}
\SetMathAlphabet{\mathtt}{normal}{\encodingdefault}{\ttdefault}{\mddefault}{\updefault}
\usepackage{microtype}

\usepackage{hyperref}
\hypersetup{colorlinks=true, urlcolor=Blue, citecolor=Green, linkcolor=BrickRed, breaklinks, unicode}

\usepackage[dvipsnames,usenames]{xcolor}
\usepackage[nocompress]{cite}
\usepackage{amsmath,mathtools,stmaryrd}
\usepackage{enumerate}

\usepackage{arydshln}
\def\Depth{\operatorname{\mathit{depth}}}

\def\Defect{\operatorname{\mathit{defect}}}
\DeclareMathOperator{\Cr}{cr}

\newtheorem{lemma}{Lemma}[section]
\newtheorem{theorem}[lemma]{Theorem}
\newtheorem{corollary}[lemma]{Corollary}

\numberwithin{figure}{section}

\begin{document}
\begin{titlepage}

\title{Lower Bounds for Electrical Reduction on Surfaces%
\thanks{This work was partially supported by NSF grant CCF-1408763.  We also greatly appreciate the support from Labex Bezout.  The conference version of the paper appears in the Proceedings of the 35th International Symposium on Computational Geometry (SoCG 2019).}}

\author{Hsien-Chih Chang\thanks{
		Department of Computer Science, Duke University, USA.  This work was initiated when the author was affiliated with University of Illinois at Urbana-Champaign.}
		\and
		Marcos Cossarini\thanks{
		Laboratoire d'Analyse et de Mathématiques Appliquées, Université Paris-Est Marne-la-Vallée, France.  This work was initiated when the author was affiliated with Instituto de Matemática Pura e Aplicada, Brazil.}
		\and
		Jeff Erickson\thanks{
		Department of Computer Science, University of Illinois at Urbana-Champaign, USA.
		}}

\date{March 22, 2019}

\maketitle

\begin{abstract}
We strengthen the connections between \emph{electrical transformations} and \emph{homotopy} from the planar setting---observed and studied since Steinitz---to arbitrary surfaces with punctures.
As a result, we improve our earlier lower bound on the number of electrical transformations required to reduce an $n$-vertex graph on surface in the worst case [SOCG 2016] in two different directions.
Our previous $\Omega(n^{3/2})$ lower bound applies only to \emph{facial} electrical transformations on plane graphs with \emph{no terminals}.
First we provide a stronger $\Omega(n^2)$ lower bound when the planar graph has two or more terminals, which follows from a quadratic lower bound on the number of homotopy moves in the annulus.
Our second result extends our earlier $\Omega(n^{3/2})$ lower bound to the wider class of \emph{planar} electrical transformations, which preserve the planarity of the graph but may delete cycles that are not faces of the given embedding.
This new lower bound follows from the observation that the \emph{defect} of the medial graph of a planar graph is the same for all its planar embeddings.
\end{abstract}

\setcounter{page}{0}
\thispagestyle{empty}
\end{titlepage}

\pagestyle{myheadings}
\markboth{Hsien-Chih Chang, Marcos Cossarini, and Jeff Erickson}
		{Lower Bounds for Planar Electrical Reduction on Surface Graphs}


\section{Introduction}

Consider the following set of local operations performed on any graph:
\begin{itemize}\itemsep=0pt
\item \emph{Leaf contraction}: Contract the edge incident to a vertex of degree $1$.
\item \emph{Loop deletion}: Delete the edge of a loop.
\item \emph{Series reduction}: Contract either edge incident to a vertex of degree $2$.
\item \emph{Parallel reduction}: Delete one of a pair of parallel edges.
\item \emph{$\arc{Y}{\Delta}$ transformation}: Delete a degree-$3$ vertex and connect its neighbors with three new~edges.
\item \emph{$\arc{\Delta}{Y}$ transformation}: Delete edges of a $3$-cycle and join its vertices to a new~vertex.
\end{itemize}
These operations and their inverses, which we call \EMPH{electrical transformations} following Colin de Verdière \etal~\cite{cgv-rep-96}, have been used for over a century to analyze electrical networks \cite{k-etscn-1899}.
Steinitz \cite{s-pr-1916,sr-vtp-34} proved that any planar network can be reduced to a single vertex using these operations.
Several decades later, Epifanov \cite{e-rpges-66} proved that any planar graph with two special vertices called \emph{terminals} can be similarly reduced to a single edge between the terminals; simpler algorithmic proofs of Epifanov's theorem were later given by Feo \cite{f-erpns-85}, Truemper \cite{t-drpg-89,t-md-92}, and Feo and Provan \cite{fp-dtert-93}.
These results have since been extended to planar graphs with more than two terminals \cite{g-dtaa-91,gs-tdrpg-11,acgp-frpwg-00,dm-ftpdw-15}
and to some families of non-planar graphs \cite{g-dtaa-91,w-drag-15}.
See Chang's thesis~\cite{c-tcgs-18} for a history of the problem. 

Despite decades of prior work, the complexity of the reduction process is still poorly understood.  Steinitz's proof implies that $O(n^2)$ electrical transformations suffice to reduce any $n$-vertex planar graph to a single vertex; Feo and Provan's algorithm reduces any 2-terminal planar graph to a single edge in $O(n^2)$ steps.
While these are the best upper bounds known, several authors have conjectured that they can be improved \cite{g-dtaa-91,fp-dtert-93,acgp-frpwg-00}.  Without any restrictions on which transformations are permitted, the only known lower bound is the trivial $\Omega(n)$.
However, Chang and Erickson recently proved that if all transformations are required to be \emph{facial}, meaning any deleted cycle must be a face of the given embedding, then reducing a plane graph without terminals to a single vertex requires $\Omega(n^{3/2})$ steps in the worst case~\cite{tangle}.
This is obtained by studying the relation between facial electrical transformations and \emph{homotopy moves}, a set of operations performed on the medial graph of the input.

In this paper, we extend our earlier lower bound for electrical transformations in two directions.
To this end, first we study multicurves on surfaces under electrical and homotopy moves; multicuves are in one-to-one correspondence with medial graphs of graph embeddings.
Specifically, in Section~\ref{S:electric} we prove that the set of \emph{tight} multicurves under electrical moves and under homotopy moves is identical.  As a consequence, any surface-embedded graph can be reduced without ever increasing its number of edges.
Previously such property is only known to hold for plane graphs \cite{nw-kg-00,tangle}.

Next, we consider plane graphs with two terminals.  In this setting, leaf deletions, series reductions, and $Y\arcto\Delta$ transformations that delete terminals are forbidden.
We prove in Section~\ref{S:terminals} that $\Omega(n^2)$ facial electrical transformations are required in the worst case to reduce a 2-terminal plane graph \emph{as much as possible}.  Not every 2-terminal plane graph can be reduced to a single edge between the terminals using only facial electrical transformations.
However, we show that any 2-terminal plane graph can be reduced to a unique minimal graph called a \emph{bullseye} using a finite number of facial electrical transformations.
Our lower bound ultimately relies on a recent $\Omega(n^2)$ lower bound on the number of homotopy moves required to tighten a contractible closed curve in the annulus \cite{untangle}.

In Section \ref{S:planar}, we consider a wider class of electrical transformations that preserve the planarity of the graph, but are not necessarily facial.
Our second main result is that $\Omega(n^{3/2})$ \emph{planar} electrical transformations are required to reduce a planar graph (without terminals) to a single vertex in the worst case.
Like our earlier lower bound for \emph{facial} electrical transformations, our proof ultimately reduces to the study of a certain curve invariant, called the \emph{defect}, of the medial graph of a given \emph{unicursal} plane graph $G$.
A key step in our new proof is the following surprising observation: Although the definition of the medial graph of $G$ depends on the embedding of $G$,
the defect of the medial graph is the same for all planar embeddings of $G$.



\section{Background}
\label{S:background}

\subsection{Types of electrical transformations}

We distinguish between three increasingly general types of electrical transformations in plane graphs: \emph{facial}, \emph{crossing-free}, and \emph{arbitrary}.  (For ease of presentation, we assume throughout the paper that plane graphs are actually embedded on the \emph{sphere} instead of the plane.)

An electrical transformation in a graph $G$ embedded on a surface $\Sigma$ is \EMPH{facial} if any deleted cycle is a face of $G$.  All leaf contractions, series reductions, and $\arc{Y}{\Delta}$ transformations are facial, but loop deletions, parallel reductions, and $\arc{\Delta}{Y}$ transformations may not be facial.
Facial electrical transformations form three dual pairs, as shown in Figure \ref{F:elec-dual}; for example, any series reduction in $G$ is equivalent to a parallel reduction in the dual graph~$G^*$.

\begin{figure}[htb]
\centering
\includegraphics[width=0.75\textwidth]{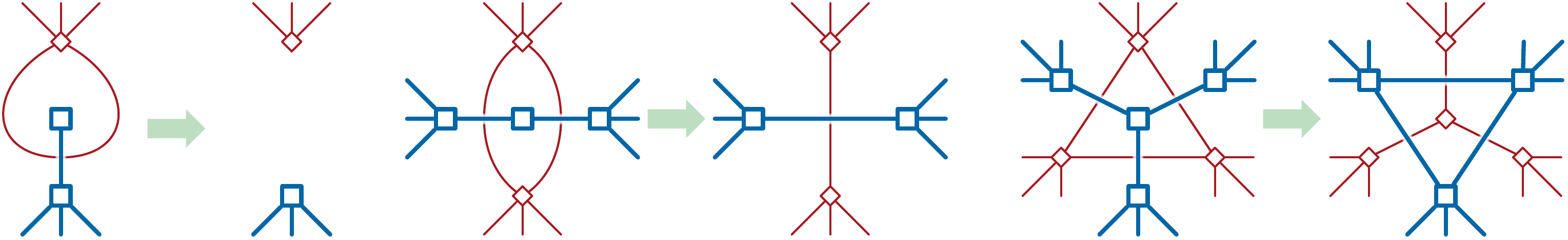}
\caption{Facial electrical transformations in a plane graph $G$ and its dual $G^*$.}
\label{F:elec-dual}
\end{figure}

An electrical transformation in $G$ is \EMPH{crossing-free} if it preserves the embeddability of the underlying graph into the same surface.  Equivalently, an electrical transformation is crossing-free if the vertices of the cycle deleted by the transformation are all incident to a common face of $G$.  All facial electrical transformations are trivially crossing-free, as are all loop deletions and parallel reductions.
If the graph embeds in the plane, crossing-free electrical transformations are also called \EMPH{planar}.
%
The only non-crossing-free electrical transformation is a $\arc{\Delta}{Y}$ transformation whose three vertices are \emph{not} incident to a common face; any such transformation introduces a $K_{3,3}$-minor into the graph, connecting the three vertices of the $\Delta$ to an interior vertex, an exterior vertex, and the new~$Y$ vertex.

\begin{figure}[ht]
\centering
\includegraphics[scale=0.3]{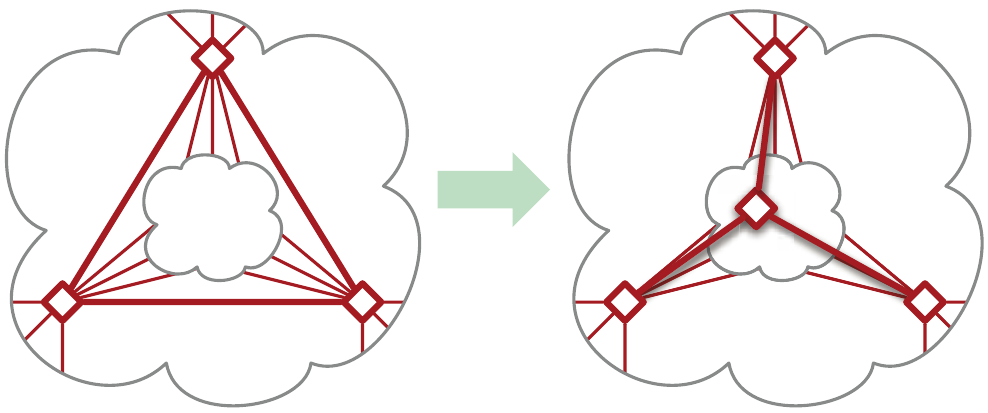}
\caption{A non-planar $\arc{\Delta}{Y}$ transformation.}
\end{figure}

\subsection{Multicurves and medial graphs}

A \EMPH{surface} is a 2-manifold with or without punctures.
Formally, a \EMPH{closed curve} in a surface $\Sigma$ is a continuous map $\gamma \colon S^1 \to \Sigma$.
A~closed curve is \EMPH{simple} if it is injective.  A \EMPH{multicurve} is a collection of one or more closed curves.
We consider only \emph{generic} multicurves, which are injective except at a finite number of (self-)intersections, each of which is a transverse double point.
A multicurve is \EMPH{connected} if its image in the surface is connected.
The image of any (non-simple) multicurve has a natural structure as a 4-regular map, whose \EMPH{vertices} are the self-intersection points of the curves, \EMPH{edges} are maximal subpaths between vertices, and \EMPH{faces} are components of the complement of the curves in the surface.
We do not distinguish between multicurves whose images are combinatorially equivalent maps.

The \EMPH{medial graph $G^\times$} of an embedded graph $G$ is another embedded graph whose vertices correspond to the edges of $G$,
and two vertices of $G^\times$ are connected by an edge if the corresponding edges in~$G$ are consecutive in cyclic order around some vertex, or equivalently, around some face in~$G$.  Every vertex in every medial graph has degree $4$; thus, every medial graph is the image of a multicurve.  Conversely, image of a non-simple multicuvre is the medial graph of some surface-embedded graph if the faces of the multicurve can be two-colored; in particular, when the surface is a sphere, the image of every non-simple multicurve is the medial graph of some plane graph.
We call an embedded graph $G$ \EMPH{unicursal} if its medial graph $G^\times$ is the image of a single closed curve.

\EMPH{Smoothing} a multicurve $\gamma$ at a vertex $x$ replaces the intersection of $\gamma$ with a small neighborhood of $x$ with two disjoint simple paths, so that the result is another 4-regular embedded graph.
There are two possible smoothings at each vertex.  More generally, a \EMPH{smoothing} of $\gamma$ is any multicurve obtained by smoothing a subset of its vertices.
For any embedded graph $G$, the smoothings of the medial graph $G^\times$ are precisely the medial graphs of minors of $G$.

\begin{figure}[ht]
\centering
\includegraphics[scale=0.25]{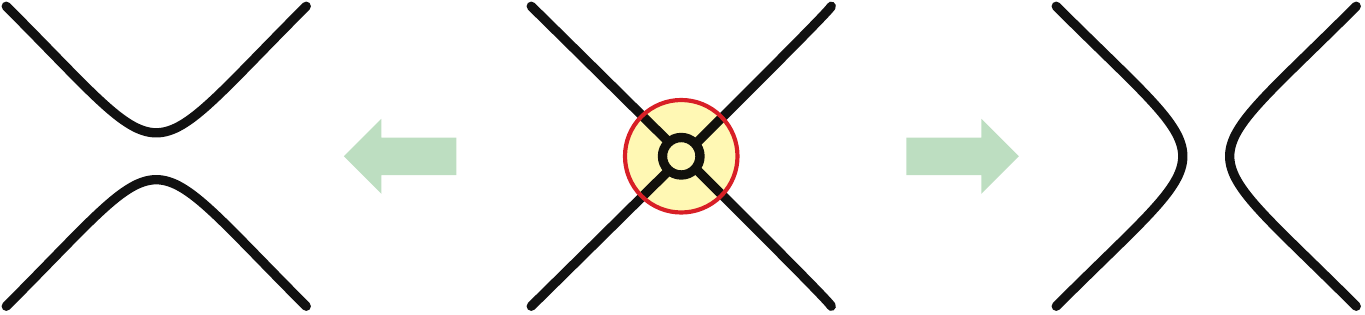}
\caption{Two possible smoothings of a vertex.}
\label{F:smoothing}
\end{figure}

\subsection{Local moves}

A \EMPH{homotopy} between two curves $\gamma$ and $\gamma'$ on the same surface $\Sigma$ is a continuous deformation from one curve to the other, formally defined as a continuous function $H\colon {S^1 \times [0,1] \to \Sigma}$ such that $H(\cdot,0) = \gamma$ and $H(\cdot,1) = \gamma'$.  
The definition of homotopy extends naturally to multicurves.  Classical topological arguments imply that two multicurves are homotopic if and only if one can be transformed into the other by a finite sequence of \EMPH{homotopy moves} (shown in Figure \ref{F:homotopy}).
%
Notice that a $\arc10$ move is applied to an empty \EMPH{loop}, and a $\arc20$ move is applied on an empty \EMPH{bigon}.
A multicurve is \EMPH{homotopically tight} (or \EMPH{h-tight} for short) if no sequence of homotopy moves leads to a multicurve with fewer vertices.

\begin{figure}[ht]
\centering
\includegraphics[width=0.75\linewidth]{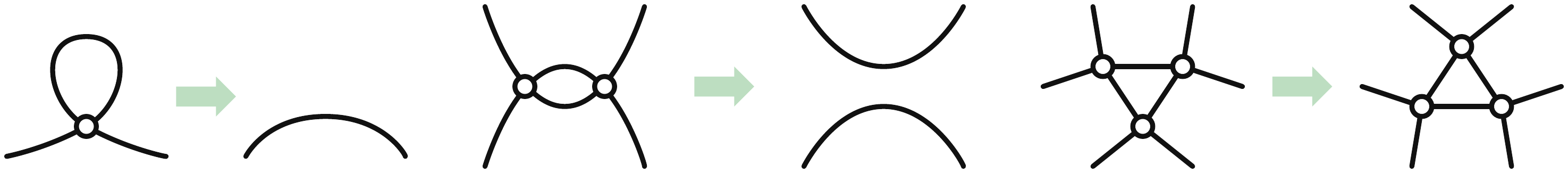}
\caption{Homotopy moves $\arc10$, $\arc20$, and $\arc33$.}
\label{F:homotopy}
\end{figure}

\begin{figure}[ht]
\centering
\includegraphics[width=0.75\textwidth]{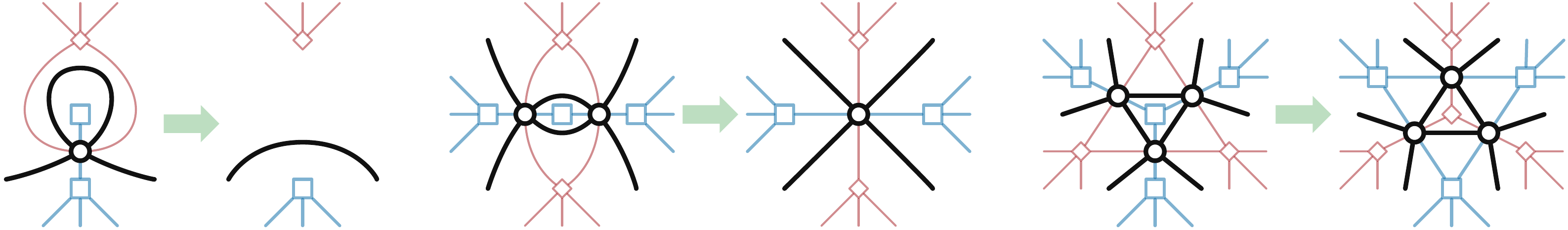}
\caption{Electrical moves $\arc10$, $\arc21$, and $\arc33$.}
\label{F:medial-elec}
\end{figure}

Facial electrical transformations in any embedded graph $G$ correspond to local operations in the medial graph~$G^\times$ that closely resemble homotopy moves.  We call these \EMPH{$\arc10$}, \EMPH{$\arc21$}, and \EMPH{$\arc33$} moves, where the numbers before and after each arrow indicate the number of local vertices before and after the move.
We collectively refer to these operations and their inverses as \EMPH{electrical moves}.
A multicurve is \EMPH{electrically tight} (or \EMPH{e-tight} for short) if no sequence of electrical moves leads to another multicurve with fewer vertices.
For multicurves on surfaces with boundary, both homotopy moves and electrical moves performed on boundary faces are forbidden.
The fact that we use same name \emph{tight} for both homotopy moves and electrical moves is not a coincidence; we will justify its usage in Section~\ref{SS:eq-tight}.

%

\section{Connection between electrical and homotopy moves}
\label{S:electric}

%
For any connected multicurve (or 4-regular embedded graph) $\gamma$ on surface $\Sigma$,
\begin{itemize}
\item let \EMPH{$X(\gamma)$} denote the minimum number of electrical moves required to tighten $\gamma$,
\item let \EMPH{$H^\downarrow(\gamma)$} denote the minimum number of homotopy moves required to tighten $\gamma$, without ever increase the number of vertices; that is, no $\arc01$ and $\arc02$ moves are allowed.
\item let \EMPH{$H(\gamma)$} denote the minimum number of homotopy moves required to tighten $\gamma$.
\end{itemize}


It is not immediately obvious
whether a multicurve $\gamma$ that is tight under monotonic homotopy moves could be further tightened by allowing $\arc 01$ and $\arc 02$ moves or not.
%
Hass and Scott~\cite{hs-scs-94} and de Graaf and Schrijver \cite{gs-mcmcr-97} independently proved that any multicurve $\gamma$ can be tightened using monotonic homotopy moves, which implies that $H^\downarrow(\gamma) = 0$ if and only if $H(\gamma) = 0$.  In other words, (standard) homotopy moves and monotonic homotopy moves share the same set of tight multicurves.
Now $H^\downarrow(\gamma) \ge H(\gamma)$ follows for any multicurve $\gamma$.

\subsection{Smoothing lemma}
\label{SS:smoothing}

We would like to compare $X(\gamma)$ with $H^\downarrow(\gamma)$ and $H(\gamma)$.
The following key lemma follows from close reading of proofs by Truemper~\cite[Lemma~4]{t-drpg-89} and several others~\cite{g-dtaa-91,nt-aafts-96,acgp-frpwg-00,nw-kg-00} that every minor of a {$\Delta$Y}-reducible graph is also {$\Delta$Y}-reducible.
%
A proof to some special cases at the level of medial curves can be found in de Graaf~\cite[Proposition~5.1]{g-gcs-94phdthesis}.
%
For the sake of completeness, we include a proof in Appendix~\ref{S:smoothing-proof}.

\begin{lemma}[Chang and Erickson~{\cite[Lemma~3.1]{tangle}}]
\label{L:smoothing-case}
Let $\gamma$ be any connected multicurve on surface $\Sigma$, and let $\check\gamma$ be a connected smoothing of $\gamma$.
Applying any sequence of $N$ electrical moves to $\gamma$ to obtain $\gamma'$.
Then one can apply a similar sequence of electrical moves of length at most $N$ to $\check\gamma$ to obtain a (possibly trivial) connected smoothing ${\check\gamma}'$ of $\gamma'$.
\end{lemma}

As a remark, using similar argument one can recover a result by Newmann-Coto \cite{n-csgs-01}: any homotopy from multicurve $\gamma$ to another multicurve $\gamma'$ that never removes vertices can be turned into a homotopy from a smoothing of $\gamma$ to a smoothing of $\gamma'$.
Chambers and Liokumovich \cite{cl-chidh-14} studied a similar problem where one wants to convert a homotopy between two \emph{simple} curves on surface into an \emph{isotopy}, without increasing the length of any intermediate curve by too much.
They showed that the desired isotopy can be obtained from a clever Euler-tour argument on the graph of all possible complete smoothings of the intermediate curves.

\medskip

Using Lemma~\ref{L:smoothing-case} one can show that $X(\gamma) \ge H^\downarrow(\gamma)$ for every planar curve $\gamma$,
a result implicit in the work of Noble and Welsh~\cite{nw-kg-00} and formally proved by Chang and Erickson~\cite{tangle}.

\begin{lemma}[Smoothing Lemma~{\cite{tangle}}]
\label{L:smoothing}
$X(\check\gamma) \le X(\gamma)$ for every connected smoothing $\check\gamma$ of every connected multicurve $\gamma$ in the plane.
\end{lemma}


\begin{lemma}[Monotonicity Lemma~{\cite{tangle}}]
\label{L:monotonicity}
For every connected multicurve $\gamma$, there is a minimum-length sequence of electrical moves that simplifies $\gamma$ to a simple closed curve that does not contain $\arc01$ or $\arc12$ moves.
\end{lemma}


\begin{lemma}[Electrical-Homotopy Inequality~{\cite{tangle}}]
\label{L:homotopy}
$X(\gamma) \ge H^\downarrow(\gamma)$ for every planar curve~$\gamma$.
\end{lemma}



\subsection{Equivalence of tightness}
\label{SS:eq-tight}

One of the main obstacles to generalize Lemmas~\ref{L:smoothing}, \ref{L:monotonicity}, and \ref{L:homotopy} to curves on arbitrary surface is that again we do not know \emph{a priori} whether the set of tight multicurves under electrical moves is the same as those under homotopy moves.
Such problem did not exist in the planar setting as all planar multicurves can be tightened to simple curves using either electrical or homotopy moves.
%
We first show that every electrically tight multicurve is also homotopically tight.

\begin{lemma}
\label{L:electric-homotopy}
Let $\gamma$ be a connected multicurve on an arbitrary surface $\Sigma$.  If $\gamma$ is electrically tight, then $\gamma$ is homotopically tight.
\end{lemma}

\begin{proof}
Let $\gamma$ be a connected multicurve in some arbitrary surface, and suppose~$\gamma$ is not homotopically tight.  Results of Hass and Scott~\cite{hs-scs-94} and de Graaf and Schrijver~\cite{gs-mcmcr-97} imply that $\gamma$ can be tightened by a finite sequence of homotopy moves that never increases the number of vertices.
In particular, applying some finite sequence of $\arc{3}{3}$ moves to $\gamma$ creates either an empty loop, which can be removed by a $\arc10$ move, or an empty bigon, which can be removed by either a $\arc20$ move or a $\arc21$ move.  Thus, $\gamma$ is not electrically tight.
\end{proof}

However, for the reverse direction, we don't have a similar monotonicity result for electrical moves on arbitrary surfaces.
A careful reading of the sequence of work by de Graaf and Schrijver~\cite{s-hcscs-89,s-dgshct-91,s-otuok-92,s-cget-92,gs-chscc-95,gs-dgs-97,gs-mcmcr-97} leads to a five-way equivalence that shows the two versions of tightness coincide when the given curve is \emph{primitive}.
Unfortunately their results do not generalize as some of the equivalences break down with the presence of non-primitive counterexamples.
See Appendix~\ref{SS:elec-homo-reduced} for more details.


\paragraph{Routing set.}
Inspired by the routing problem studied by de Graaf and Schrijver~\cite{gs-dgs-97}, we introduce the notion of \emph{routing set}.  Despite its na\"ive look, the routing set satisfies a crucial property that encapsulates the whole difficulty of the problem, which allows us to bypass the heavy machinery developed for the primitive case.
We then use the established equivalence of tightness to derive the monotonicity lemma for electrical moves on arbitrary multicurves.

\def\Route{\operatorname{\mathit{route}}}

For any multicurve $\gamma$, the \EMPH{routing set} of $\gamma$ is the following collection of homotopy classes:
\[
\EMPH{$\Route(\gamma)$} \coloneqq \Set{\Big. \Brack{\check\gamma} \mid \text{$\check\gamma$ is a smoothing of $\gamma$}}.
\]
Each homotopy class in $\Route(\gamma)$ is referred as a \EMPH{route} of $\gamma$.

\begin{lemma}
\label{L:route-invariant}
Routing set of $\gamma$ is invariant under electrical moves for any multicurve $\gamma$.
\end{lemma}

\begin{proof}

Let $\gamma'$ be the multicurve obtained from performing one electrical move to $\gamma$.
Because electrical moves are closed under inverses, we only need to prove that $\Route(\gamma) \subseteq \Route(\gamma')$.

Let $\check\gamma$ be an arbitrary smoothing of $\gamma$; $[\check\gamma]$ is in $\Route(\gamma)$ by definition.
By Lemma~\ref{L:smoothing-case}, one can obtain a smoothing ${\check\gamma}'$ of $\gamma'$ that is at most one electrical move away from $\check\gamma$.%
\footnote{Although Lemma~\ref{L:smoothing-case} is stated with respect to \emph{connected} smoothings, the proof of the lemma (see Appendix~\ref{S:smoothing-proof}) reveals that similar statement holds for arbitrary smoothings by allowing an additional $\arc 00$ move that creates/contracts simple cycles.  In particular, such move does not change the homotopy class of a multicurve.}
In particular, $[{\check\gamma}']$ is in $\Route(\gamma')$.
If ${\check\gamma}'$ is equal to $\check\gamma$ or is obtained from $\check\gamma$ using a $\arc 10$, $\arc 01$, or $\arc 33$ move, then immediately we have $[{\check\gamma}] = [{\check\gamma}']$ to be a route in $\Route(\gamma')$.
If ${\check\gamma}'$ is obtained from $\check\gamma$ using a $\arc 21$ move, consider the multicurve ${\check\gamma}^\circ$ obtained from $\check\gamma$ by performing a $\arc 20$ move (on the same empty bigon) instead.
${\check\gamma}^\circ$ is a smoothing of ${\check\gamma}'$,
which in turn is a smoothing of $\gamma'$.
Because $\arc 20$ is a homotopy move, $[\check\gamma] = [{\check\gamma}^\circ]$ is a route in $\Route(\gamma')$.
Similarly when ${\check\gamma}'$ is obtained from $\check\gamma$ using a $\arc 12$ move, we consider $\check\gamma$ as a smoothing of ${\check\gamma}'$ thus $[\check\gamma]$ is a route in $\Route(\gamma')$.
This concludes the proof.
\end{proof}

The \EMPH{intersection number} of a homotopy class $[\gamma]$ is defined to be the minimum number of vertices among all curves homotopic to $\gamma$.
The \EMPH{main routes} of $\gamma$ are those routes of $\gamma$ that achieve the maximum intersection number.

%

\begin{lemma}
\label{L:eq-tight}
Any homotopically tight multicurve is also electrically tight.
\end{lemma}

\begin{proof}
Assume for contradiction that there is an h-tight multicurve $\gamma$ that is not e-tight.  Tighten $\gamma$ using electrical moves to an e-tight multicurve $\gamma'$ with less number of vertices than $\gamma$.
Now by Lemma~\ref{L:route-invariant} the routing set of $\gamma$ and $\gamma'$ is the same; in particular, $[\gamma']$ is a main route of both $\gamma$ and $\gamma'$.
However since both $\gamma$ and $\gamma'$ are h-tight, the intersection number of $[\gamma]$ is strictly greater than the intersection number of $[\gamma']$ and thus $[\gamma']$ cannot be a main route of $\gamma$, a contradiction.
\end{proof}

\subsection{Monotonicity of electrical moves}
\label{SS:monotone-electric}

As a corollary of Lemma~\ref{L:eq-tight}, we are ready to generalize
the monotonicity lemma (Lemma~\ref{L:monotonicity}) to multicurves on general surfaces.

\begin{lemma}
\label{L:smoothing-surface}
Let $\gamma$ be any connected multicurve $\gamma$ on surface $\Sigma$, and let $\check\gamma$ be a connected smoothing of~$\gamma$, satisfying $\Route(\gamma) = \Route(\check\gamma)$.  Then $X(\check\gamma) \le X(\gamma)$ holds.
\end{lemma}

\begin{proof}
Let $\gamma$ be a connected multicurve with $n(\gamma)$ vertices, and let $\check\gamma$ be a connected smoothing of $\gamma$.  If $X(\gamma)$ equals to zero, then $\gamma$ is both e-tight and h-tight by Lemma~\ref{L:electric-homotopy}.
The fact that $\Route(\gamma) = \Route(\check\gamma)$ implies that $[\gamma]$ is a route of $\check\gamma$ and its intersection number is equal to $n(\gamma)$.  If $\check\gamma$ is a proper smoothing of $\gamma$, then the intersection number of any route of $\check\gamma$ is strictly less then $n(\gamma)$, a contradiction.  As a result, the only smoothing of $\gamma$ satisfying the condition is $\gamma$ itself, and therefore the inequality trivially holds.

Otherwise, applying a minimum-length sequence of electrical moves that tightens $\gamma$.
By Lemma~\ref{L:smoothing-case} there is another sequence of electrical moves of length at most $X(\gamma)$ that tightens $\check\gamma$.
We immediately have $X(\check\gamma) \le X(\gamma)$ and the lemma is proved.
\end{proof}

\begin{lemma}
\label{L:monotonicity-surface}
For any connected multicurve $\gamma$, there is a minimum-length sequence of electrical moves that tightens $\gamma$ that does not contain $\arc01$ or $\arc12$ moves.
\end{lemma}

The proof follows almost verbatim from Lemma~\ref{L:monotonicity} after substituting Lemma~\ref{L:smoothing-surface} for Lemma~\ref{L:smoothing} and applying Lemma~\ref{L:route-invariant}.

\begin{proof}
Consider a minimum-length sequence of electrical moves that tights $\gamma$.
For any integer $i \ge 0$, let $\gamma_i$ denote the result of the first $i$ moves in this sequence.
Minimality of the tightening sequence implies that
$X(\gamma_i)$ decreases as $i$ grows.
Now let $i$ be an arbitrary index such that~$\gamma_i$ is obtained from performing a $\arc01$ or $\arc 12$ move on $\gamma_{i-1}$. Then $\gamma_{i-1}$ is a connected proper smoothing of $\gamma_i$, and by Lemma~\ref{L:route-invariant}, $\Route(\gamma_{i-1}) = \Route(\gamma_{i})$ holds.
Now Lemma~\ref{L:smoothing-surface} implies that $X(\gamma_{i-1}) \le X(\gamma_i)$, a contradiction.
\end{proof}

\section{Two-terminal plane graphs}
\label{S:terminals}

Most applications of electrical reductions, starting with Kennelly's computation of effective resistance~\cite{k-etscn-1899}, designate two vertices of the input graph as \emph{terminals} and require a reduction to a single edge between those terminals.  In this context, electrical transformations that delete either of the terminals are forbidden; specifically:
leaf contractions when the leaf is a terminal,
series reductions when the degree-2 vertex is a terminal, and
$\arc{Y}{\Delta}$ transformations when the degree-3 vertex is a terminal.
%
%
An important subtlety here
is that not every 2-terminal planar graph can be reduced to a single edge using only \emph{facial} electrical transformations.  The simplest bad example is the three-vertex graph shown in Figure~\ref{F:bad-plane-graph}.

\begin{figure}[ht]
\centering
\includegraphics[scale=0.3]{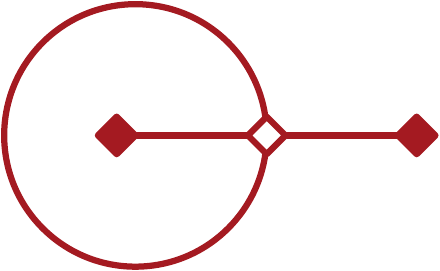}
\caption{A facially irreducible 2-terminal plane graph; solid vertices are the terminals.}
\label{F:bad-plane-graph}
\end{figure}

In this section, we show that in the worst case, $\Omega(n^2)$ facial electrical transformations are required to reduce a 2-terminal plane graph with $n$ vertices \emph{as much as possible}.
The medial graph $G^\times$ of any 2-terminal plane graph $G$ is properly considered as a multicurve embedded in the annulus; the faces of $G^\times$ that correspond to the terminals are removed from the surface.
%
The main strategy is to lower bound $X(G^\times)$ by some function of $H(G^\times)$, then defer to the quadratic lower bound for untangling annular curve using homotopy moves \cite{untangle}.
To this end, we generalize Lemma~\ref{L:homotopy} to annular curves; such result is obtained by the understanding of tight multicurves on the annulus.

First, we prove in Section \ref{SS:tight-annulus} that any annular curve can be tightened to a unique family of curves.
Next in Section~\ref{SS:smoothing-annulus}, we generalize the results by Chang and Erickson~\cite{tangle}, in particular the electrical-homotopy inequality (Lemma~\ref{L:homotopy}), to the annular case.
We prove our quadratic lower bound in Section~\ref{SS:quad-lower}.
Existing algorithms for reducing an arbitrary 2-terminal plane graphs to a single edge rely on an additional operation which we call a \emph{terminal-leaf contraction}, in addition to facial electrical transformations.  We discuss this subtlety in more detail in Section \ref{SS:terminal-leaf}.

\subsection{Tight annular curves}
\label{SS:tight-annulus}

The \EMPH{winding number} of a directed closed curve $\gamma$ in the annulus is the number of times any generic path $\pi$ from one (fixed) boundary component to the other crosses $\gamma$ from left to right, minus the number of times $\pi$ crosses~$\gamma$ from right to left.  Two directed closed curves in the annulus are homotopic if and only if their winding numbers are equal.

The \EMPH{depth} of any multicurve $\gamma$ in the annulus is the minimum number of times a path from one boundary to the other crosses~$\gamma$; thus, depth is essentially an unsigned version of winding number.
%
Just as the winding number around the boundaries is a complete homotopy invariant for curves in the annulus, the depth turns out to be a complete invariant for electrical moves on the annular multicurves.

\begin{lemma}
\label{L:depth}
Electrical moves do not change the depth of any annular multicurve.
\end{lemma}



For any integer $d>0$, let \EMPH{$\alpha_d$} denote the unique closed curve in the annulus with $d-1$ vertices and winding number~$d$.  Up to isotopy, this curve can be parametrized in the plane as
\[
	\alpha_d(\theta)
	~:=~
	\Paren{\big. (\cos(\theta)+2)\cos(d\theta),~ (\cos(\theta)+2)\sin(d\theta) }.
\]
In the notation of our other papers \cite{tangle,annulus}, $\alpha_d$ is the \emph{flat torus knot} $T(d,1)$.

The following lemmas are direct consequences of Lemma~\ref{L:eq-tight}; here we provide simple proofs using only winding number and depth of annular curves.

\begin{lemma}
\label{L:reduced}
For any integer $d>0$, the curve $\alpha_d$ is both h-tight and e-tight.
\end{lemma}

\begin{proof}
Every connected multicurve in the annulus with either winding number $d$ or depth $d$ has at least $d+1$ faces (including the faces containing the boundaries of the annulus) and therefore, by Euler's formula, has at least $d-1$ vertices.
\end{proof}

\begin{lemma}
\label{L:homotopy-unique}
If $\gamma$ is an h-tight connected annular multicurve, then $\gamma = \alpha_d$ for some $d$.
\end{lemma}

\begin{proof}
A multicurve in the annulus is h-tight if and only if its constituent curves are h-tight \emph{and disjoint}.  Thus, any \emph{connected} h-tight multicurve is actually a single closed curve.  Any two curves in the annulus with the same winding number are homotopic \cite{h-udtse-35}.  Finally, up to isotopy, $\alpha_d$ is the only closed curve in the annulus with winding number $d$ and $d-1$ vertices \cite[Lemma~1.12]{hs-ics-85}.
\end{proof}


\begin{corollary}
\label{C:electric-unique}
A connected multicurve $\gamma$ in the annulus is e-tight if and only if $\gamma = \alpha_{\Depth(\gamma)}$; therefore, any annular multicurve $\gamma$ is e-tight if and only if $\gamma$ is h-tight.
\end{corollary}


%
%

\subsection{Smoothing lemma in the annulus}
\label{SS:smoothing-annulus}

Equipped with the understanding of tight annular curves, we are ready to extend the results in Section~\ref{SS:smoothing} to the annulus.

\begin{lemma}
\label{L:smoothing-annulus}
For any connected smoothing $\check\gamma$ of any connected multicurve $\gamma$ in the annulus, we have $X(\check\gamma) + \frac{1}{2}\Depth(\check\gamma) \le X(\gamma) + \frac{1}{2}\Depth(\gamma)$.
\end{lemma}

\begin{proof}
Let $\gamma$ be an arbitrary connected multicurve in the annulus, and let $\check\gamma$ be an arbitrary connected smoothing of~$\gamma$.  Without loss of generality, we can assume that $\gamma$ is non-simple, since otherwise the lemma is vacuous.

If $\gamma$ is already e-tight, then $\gamma = \alpha_d$ for some integer $d\ge 2$ by Corollary \ref{C:electric-unique}.  (The curves $\alpha_0$ and~$\alpha_1$ are simple.)
First, suppose $\check\gamma$ is a connected smoothing of $\gamma$ obtained by smoothing a single vertex $x$.
The smoothed curve~$\check\gamma$ contains a single empty loop if $x$ is the innermost or outermost vertex of $\gamma$, or a single empty bigon otherwise.  Applying one $\arc10$ or $\arc20$ move transforms $\check\gamma$ into the curve $\alpha_{d-2}$, which is e-tight by Lemma~\ref{L:reduced}.
Thus we have $X(\check\gamma) = 1$ and $\Depth(\check\gamma) = d-2$, which implies $X(\check\gamma) + \frac{1}{2}\Depth(\check\gamma) = X(\gamma) + \frac{1}{2}\Depth(\gamma)$.
As for the general case when $\check\gamma$ is obtained from $\gamma$ by smoothing more than one vertices, the statement follows from the previous case by induction on the number of smoothed vertices.

If $\gamma$ is not e-tight, applying a minimum-length sequence of electrical moves that tightens $\gamma$ into some curve~$\gamma'$.
By Lemma~\ref{L:smoothing-case} there is another sequence of electrical moves of length at most $X(\gamma)$ that tightens $\check\gamma$ to some connected smoothing $\check\gamma'$ of $\gamma'$, which can be further tightened electrically to an e-tight curve using arguments in the previous paragraph because $\gamma'$ is e-tight.
This implies that $X(\check\gamma) \le X(\gamma) + \frac{1}{2}(\Depth(\gamma') - \Depth(\check\gamma'))$.
By Lemma~\ref{L:depth}, $\gamma$ and $\gamma'$ have the same depth, and $\check\gamma$ and $\check\gamma'$ have the same depth.
Therefore $X(\check\gamma) + \frac{1}{2}\Depth(\check\gamma) \le X(\gamma) + \frac{1}{2}\Depth(\gamma)$ and the lemma is proved.
\end{proof}

\begin{lemma}
\label{L:min-monotonicity}
For every connected multicurve $\gamma$ in the annulus, there is a minimum-length sequence of electrical moves that tightens $\gamma$ to $\alpha_{\Depth(\gamma)}$ without $\arc01$ or $\arc12$ moves.
\end{lemma}

The proof follows almost verbatim from Lemma~\ref{L:monotonicity} and~\ref{L:monotonicity-surface} after substituting Lemma~\ref{L:smoothing-annulus} for Lemma~\ref{L:smoothing}.

\begin{proof}
Consider a minimum-length sequence of electrical moves that tightens an arbitrary connected multicurve~$\gamma$ in the annulus.  For any integer $i \ge 0$, let $\gamma_i$ denote the result of the first $i$ moves in this sequence.  Suppose~$\gamma_i$ has one more vertex than $\gamma_{i-1}$ for some index $i$.
Then $\gamma_{i-1}$ is a connected proper smoothing of $\gamma_i$, and $\Depth(\gamma_i) = \Depth(\gamma_{i-1})$ by Lemma~\ref{L:depth}; so Lemma~\ref{L:smoothing-annulus} implies that $X(\gamma_{i-1}) \le X(\gamma_i)$, contradicting our assumption that the reduction sequence has minimum length.
\end{proof}

\begin{lemma}
\label{L:ineq}
$X(\gamma) + \frac{1}{2}\Depth(\gamma) \ge H^\downarrow(\gamma) \ge H(\gamma)$ for every closed curve $\gamma$ in the annulus.
\end{lemma}

\begin{proof}
%
Let $\gamma$ be a closed curve in the annulus.  If $\gamma$ is already e-tight, then $X(\gamma) = H^\downarrow(\gamma) = 0$ by Lemma~\ref{L:electric-homotopy} (or Corollary~\ref{C:electric-unique}), so the lemma is trivial.
Otherwise, consider a minimum-length sequence of electrical moves that tightens $\gamma$.  By Lemma~\ref{L:min-monotonicity}, we can assume that the first move in the sequence is neither $\arc01$ nor $\arc12$.  If the first move is $1\arcto 0$ or $3\arcto 3$, the theorem immediately follows by induction on $X(\gamma)$, since by Lemma~\ref{L:depth} neither of these moves changes the depth of the curve.

The only interesting first move is $\arc21$.  Let $\gamma'$ be the result of this $\arc21$ move, and let $\gamma^\circ$ be the result if we perform the $\arc20$ move on the same empty bigon instead.  The minimality of the sequence implies $X(\gamma) = X(\gamma') + 1$, and we trivially have $H^\downarrow(\gamma) \le H^\downarrow(\gamma^\circ) + 1$.  Because~$\gamma$ is a single curve, $\gamma^\circ$ is also a single curve and therefore a connected proper smoothing of $\gamma'$.
%
%
%
Thus, Lemma~\ref{L:depth}, Lemma~\ref{L:smoothing-annulus}, and induction on the number of vertices imply
\begin{align*}
	X(\gamma) + \frac{1}{2}\Depth(\gamma)
	&~=~
	X(\gamma') + \frac{1}{2}\Depth(\gamma') + 1
\\	&~\ge~
	X(\gamma^\circ) + \frac{1}{2}\Depth(\gamma^\circ) + 1
\\	&~\ge~
	H^\downarrow(\gamma^\circ) + 1
\\	&~\ge~
	H^\downarrow(\gamma),
\end{align*}
which completes the proof.
\end{proof}

\subsection{Quadratic lower bound}
\label{SS:quad-lower}

\paragraph{Bullseyes.}
For any $k>0$, let \EMPH{$B_k$} denote the 2-terminal plane graph that consists of a path of length~$k$ between the terminals, with a loop attached to each of the $k-1$ interior vertices, embedded so that collectively they form concentric circles that separate the terminals.  We call each graph $B_k$ a \EMPH{bullseye}.  For example, $B_1$ is just a single edge; $B_2$ is shown in Figure \ref{F:bad-plane-graph}; and $B_4$ is shown on the left in Figure~\ref{F:bullseye4}.
The medial graph $B_k^\times$ of the $k$th bullseye is the curve $\alpha_{2k}$.  Because different bullseyes have different medial depths, Lemma \ref{L:depth} implies that no bullseye can be transformed into any other bullseye by facial electrical transformations.

\begin{figure}[ht]
\centering
\raisebox{-0.5\height}{\includegraphics[scale=0.2]{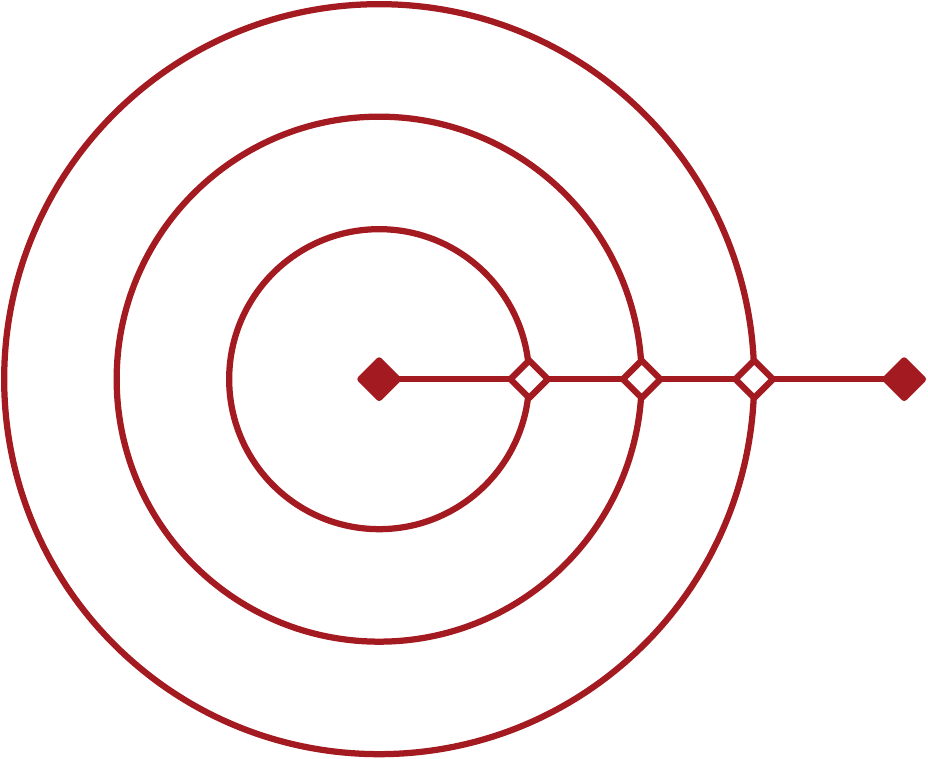}}
\qquad
\raisebox{-0.5\height}{\includegraphics[scale=0.2]{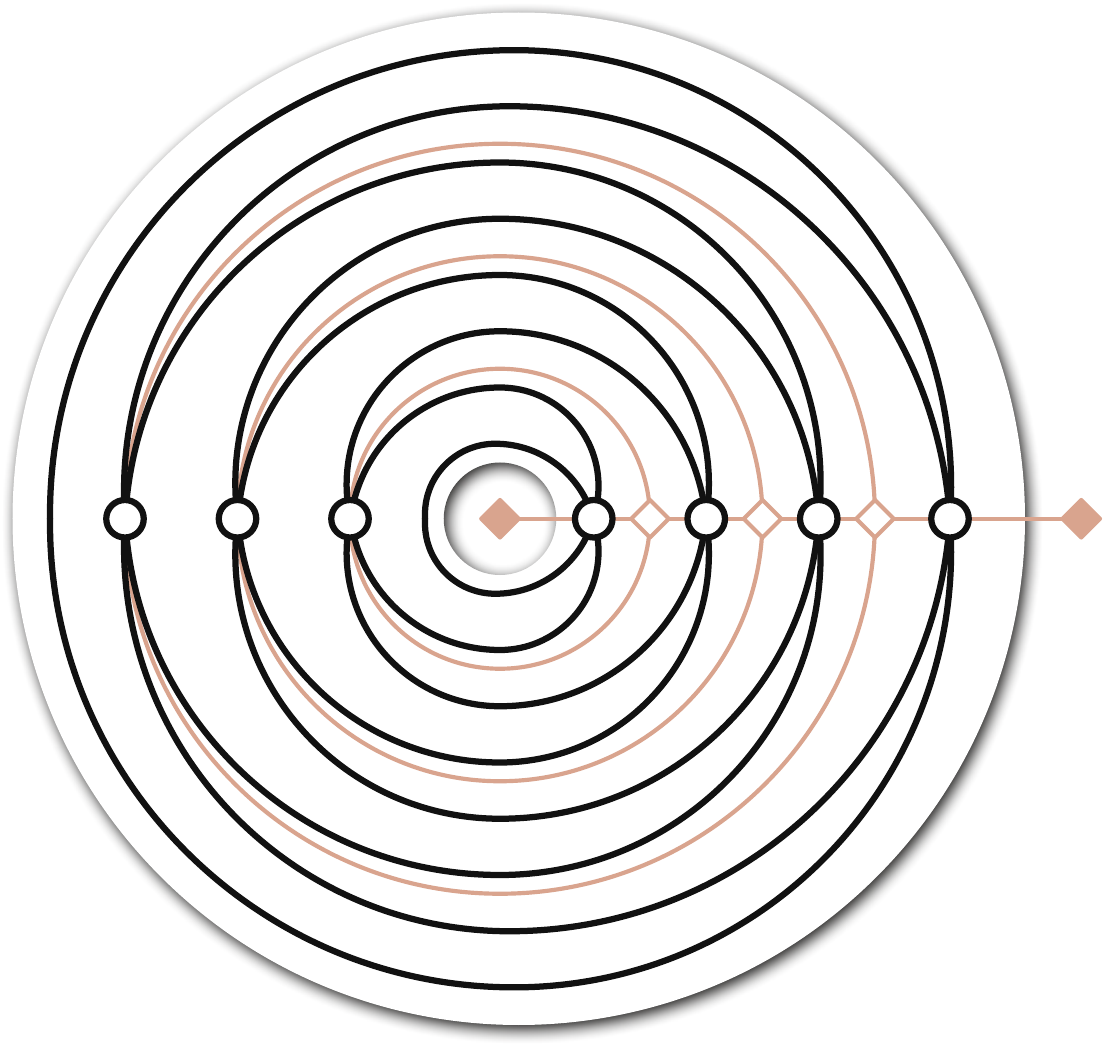}}
\qquad
\raisebox{-0.5\height}{\includegraphics[scale=0.2]{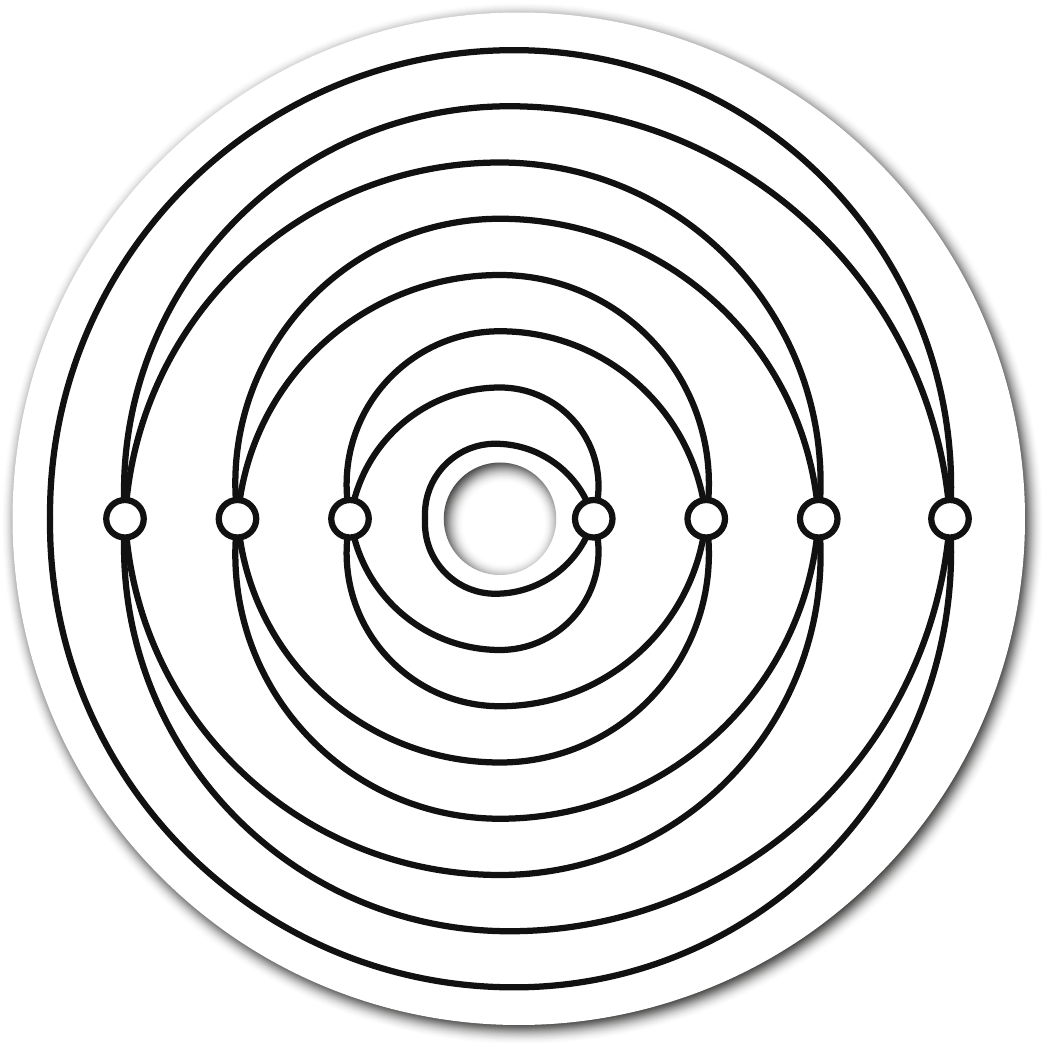}}
\caption{The bullseye graph $B_4$ and its medial graph $\alpha_8$.}
\label{F:bullseye4}
\end{figure}

The following corollary is now immediate from the electrical-homotopy inequality for annular curves (Lemma~\ref{L:ineq}).

%

\begin{theorem}
\label{Th:2-term-reduction}
Let $G$ be a 2-terminal plane graph, and let $\gamma$ be any unicursal smoothing of~$G^\times$.  Reducing~$G$ to a bullseye requires at least $H(\gamma) - \tfrac{1}{2}\Depth(\gamma)$ facial electrical transformations.
\end{theorem}

Chang~\etal~\cite{untangle} presented an infinite family of contractible curves in the annulus parametrized by their number of vertices $n$ that require $\Omega(n^2)$ homotopy moves to tighten.
Every contractible curve is the medial graph of some 2-terminal plane graph (because they have even depth and thus the faces can be two-colored~\cite{t-k1-1876}).  Euler's formula implies that every $n$-vertex curve in the annulus has exactly $n+2$ faces (including the boundary faces) and therefore has depth at most $n+1$.

\begin{corollary}
Reducing a 2-terminal plane graph to a bullseye requires $\Omega(n^2)$ facial electrical transformations in the worst case.
\end{corollary}

\subsection{Terminal-leaf contractions}
\label{SS:terminal-leaf}

The electrical reduction algorithms of Feo\cite{f-erpns-85}, Truemper \cite{t-drpg-89}, and Feo and Provan~\cite{fp-dtert-93} rely exclusively on facial electrical transformations, plus one additional operation.
\begin{itemize}
\item
\emph{Terminal-leaf contraction}: Contract the edge incident to a \emph{terminal} vertex with degree $1$.  The neighbor of the deleted terminal becomes a new terminal.
\end{itemize}
Terminal-leaf contractions are also called \emph{FP-assignments}, after Feo and Provan \cite{g-dtaa-91, gs-tdrpg-11, dm-ftpdw-15}.
Later algorithms for reducing plane graphs with three or four terminals~\cite{gs-tdrpg-11,acgp-frpwg-00,dm-ftpdw-15} also use only facial electrical transformations and terminal-leaf contractions.

Formally, terminal-leaf contractions are \emph{not} electrical transformations, as they can change the value one wants to compute.  For example, if the edges in the graph shown in Figure~\ref{F:bad-plane-graph} represent $1\Omega$ resistors, a terminal-leaf contraction changes the effective resistance between the terminals from $2\Omega$ to~$1\Omega$.
However, both Gilter \cite{g-dtaa-91} and Feo and Provan \cite{fp-dtert-93} observed that any sequence of facial electrical transformations and terminal-leaf contractions can be simulated on the fly by a sequence of \emph{planar} electrical transformations.
Specifically, we simulate the first leaf contraction at either terminal by simply marking that terminal and proceeding as if its unique neighbor were a terminal.
Later electrical transformations involving the neighbor of a marked terminal may no longer be facial, but they will still be planar; terminal-leaf contractions at the unique neighbor of a marked terminal become series reductions.  At the end of the sequence of transformations, we perform a final series reduction at the unique neighbor of each marked terminal.

Unfortunately, terminal-leaf contractions change both the depth of the medial graph and the curve invariants that imply the quadratic homotopy lower bound.  As a result, our quadratic lower bound proof breaks down if we allow terminal-leaf contractions.

\section{Planar electrical transformations}
\label{S:planar}

Finally, we extend our earlier $\Omega(n^{3/2})$ lower bound for reducing plane graphs%
---\emph{without} terminals using only facial electrical transformations---to the larger class of \emph{planar} electrical transformations.
Recall that a plane graph $G$ \emph{unicursal} if its medial graph $G^\times$ is the image of a single closed curve.
As in our earlier work \cite{tangle}, we analyze electrical transformations in an unicursal plane graph $G$ in terms of a certain invariant of the medial graph of $G$ called \EMPH{defect}, first introduced by Aicardi~\cite{a-tc-94} and Arnold~\cite{a-tipcc-94, a-pctip-94}.
Our extension to non-facial electrical transformations is based on the following surprising observation:  Although the medial graph of $G$ depends on its embedding, the \emph{defect} of the medial graph of $G$ does not.

\begin{theorem}
\label{Th:flip-defect}
Let $G$ and $H$ be planar embeddings of the same abstract planar graph.  If $G$ is unicursal, then~$H$ is unicursal and $\Defect(G^\times) = \Defect(H^\times)$.
\end{theorem}

The main goal of the section is to prove Theorem~\ref{Th:flip-defect}.

\subsection{Defect}

Let $\gamma$ be an arbitrary closed curve on the sphere.  Choose an arbitrary basepoint $\gamma(0)$ and an arbitrary orientation for $\gamma$.  For any vertex~$x$ of $\gamma$, we define $\sgn(x) = +1$ if the first traversal through~$x$ crosses the second traversal from right to left, and $\sgn(x) = -1$ otherwise.
Two vertices $x$ and~$y$ are \EMPH{interleaved}, denoted \EMPH{$x\between y$}, if they alternate in cyclic order---$x$, $y$, $x$, $y$---along $\gamma$.  Finally, following Polyak~\cite{p-icfgd-98}, we can define
\[
	\Defect(\gamma) ~\coloneqq~ -2 \sum_{x\between y} \sgn(x)\cdot\sgn(y),
\]
where the sum is taken over all interleaved pairs of vertices of $\gamma$.

Trivially, every simple closed curve has defect zero.  Straightforward case analysis \cite{p-icfgd-98} implies that the defect of a curve does not depend on the choice of basepoint or orientation.
Moreover, any homotopy move changes the defect of a curve by at most $2$; see the paper by Chang and Erickson~\cite[Section 2.1]{tangle} for an explicit case breakdown.  Defect is also preserved by any homeomorphism from the sphere
to itself, including reflection.

\subsection{Navigating between planar embeddings}

\paragraph{Short history of planar embeddings.}
A classical result of Adkisson \cite{a-ccccc-30} and Whitney \cite{w-cgcg-32} is that every 3-connected planar graph has an essentially unique planar embedding.
Mac Lane \cite{m-scpcg-37} described how to count the planar embeddings of any biconnected planar graph, by decomposing it into its triconnected components.
Stallmann \cite{s-upqtp-85,s-cpe-93} and Cai \cite{c-cepgd-93} extended Mac Lane's algorithm to arbitrary planar graphs, by decomposing them into biconnected components.  Mac Lane's decomposition is also the basis of the SPQR-tree data structure of Di Battista and Tamassia \cite{dt-ipt-89,dt-opt-96}, which encodes all planar embeddings of an arbitrary planar graph.

Whitney~\cite{w-2g-1933,t-w2tg-80} showed
that any planar embedding of a 2-connected planar graph $G$ can be transformed into any other embedding by a finite sequence of \EMPH{split reflections}, defined as follows.
A \EMPH{split curve} is a simple closed curve $\sigma$ whose intersection with the embedding of $G$ consists of two vertices $x$ and $y$; without loss of generality, $\sigma$ is a circle with $x$ and $y$ at opposite points.  A split reflection modifies the embedding of $G$ by reflecting the subgraph inside $\sigma$ across the line through $x$ and $y$.

\begin{lemma}
Let $G$ be an arbitrary 2-connected planar graph.  Any two planar embeddings of $G$ can be transformed into one other by a finite sequence of split reflections.
\end{lemma}

To navigate among the planar embeddings of \emph{arbitrary} connected planar graphs, we need two additional operations.  First, we allow split curves that intersect $G$ at only a single cut vertex; a \EMPH{cut reflection} modifies the embedding of $G$ by reflects the subgraph inside such a curve.
More interestingly, we also allow degenerate split curves that pass through a cut vertex $x$ of $G$ \emph{twice}, but are otherwise simple and disjoint from $G$.  The interior of a degenerate split curve $\sigma$ is an open topological disk.
A \EMPH{cut eversion} is a degenerate split reflection that everts the embedding of the subgraph of $G$ inside such a curve, intuitively by mapping the interior of $\sigma$ to an open circular disk (with two copies of $x$ on its boundary), reflecting the interior subgraph, and then mapping the resulting embedding back to the interior of $\sigma$.  Structural results of Stallman \cite{s-upqtp-85,s-cpe-93} and Di Battista and Tamassia \cite[Section 7]{dt-opt-96} imply the following.

\begin{figure}[ht]
\centering
\includegraphics[scale=0.25]{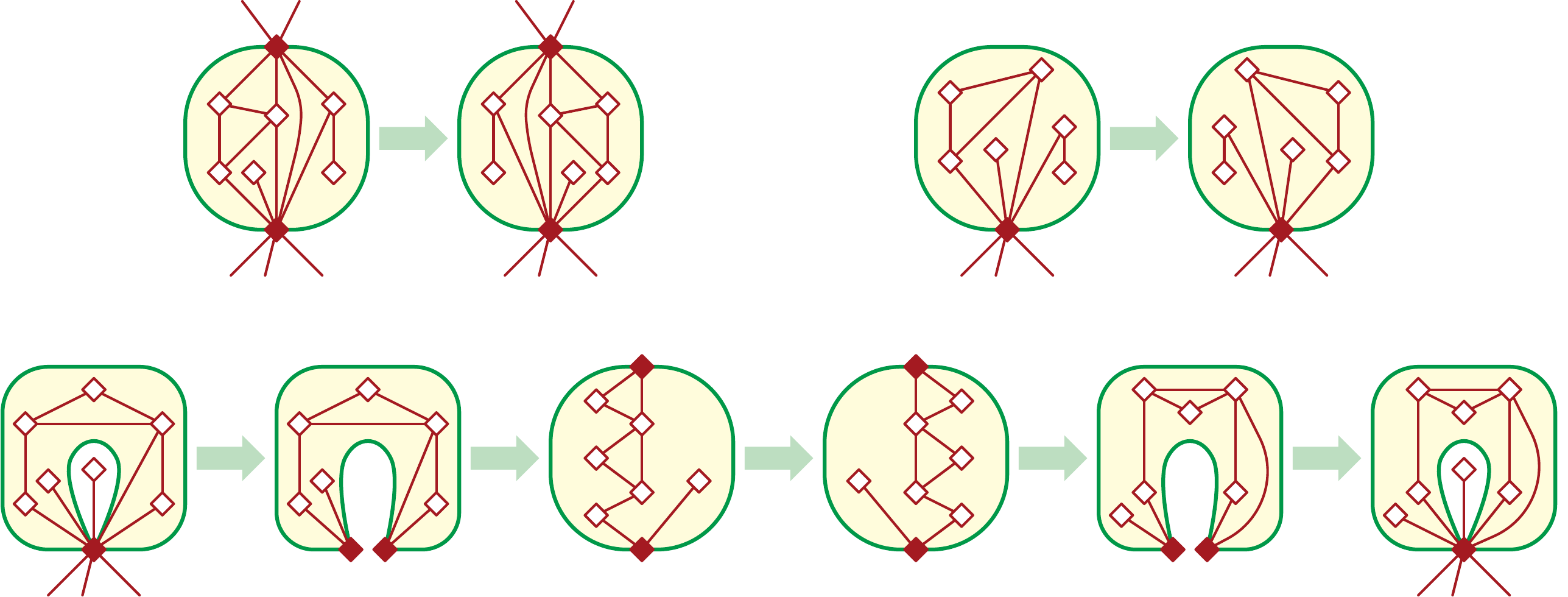}
\caption{Top row: A regular split reflection and a cut reflection.  Bottom row: a cut eversion.}
\end{figure}

\begin{lemma}
\label{L:navigate}
Let $G$ be an arbitrary connected planar graph.  Any planar embedding of $G$ can be transformed into any other planar embedding of $G$ by a finite sequence of split reflections, cut reflections, and cut eversions.
\end{lemma}

\subsection{Tangle flips}

\def\TightenIn{\Cap}
\def\TightenOut{\Cup}

Now consider the effect of the operations stated in Lemma~\ref{L:navigate} on the medial graph $G^\times$.
By assumption, $G$ is unicursal so that $G^\times$ is a single closed curve.  Let $\sigma$ be any (possibly degenerate) split curve for $G$.  Embed $G^\times$ so that every medial vertex lies on the corresponding edge in $G$, and every medial edge intersects $\sigma$ at most once.
By the Jordan curve theorem, we can assume without loss of generality that $\sigma$ is a circle, and that the intersection points $\gamma\cap\sigma$ are evenly spaced around $\sigma$.
A \EMPH{tangle} of $\gamma$ is the intersection of $\gamma$ with either disk bounded by $\sigma$; each tangle consists of one or more subpaths of $\gamma$ called \EMPH{strands}.  We arbitrarily refer to the two tangles defined by $\sigma$ as the \emph{interior} and \emph{exterior} tangles of $\sigma$.
Split curve $\sigma$ intersects at most four edges of~$G^\times$, so the tangle of $G^\times$ inside $\sigma$ has at most two strands.
Moreover, reflecting (or everting) the subgraph of $G$ inside $\sigma$ induces a \EMPH{flip} of this tangle of $G^\times$.
Any tangle can be \emph{flipped} by reflecting the disk containing it, so that each strand endpoint maps to a different strand endpoint; see Figure \ref{F:flip}.  Straightforward case analysis implies that flipping any tangle of~$G^\times$ with at most two strands transforms $G^\times$ into another closed curve; see Figure~\ref{F:flip-cases}.

\begin{figure}[ht]
\centering
\includegraphics[scale=0.3]{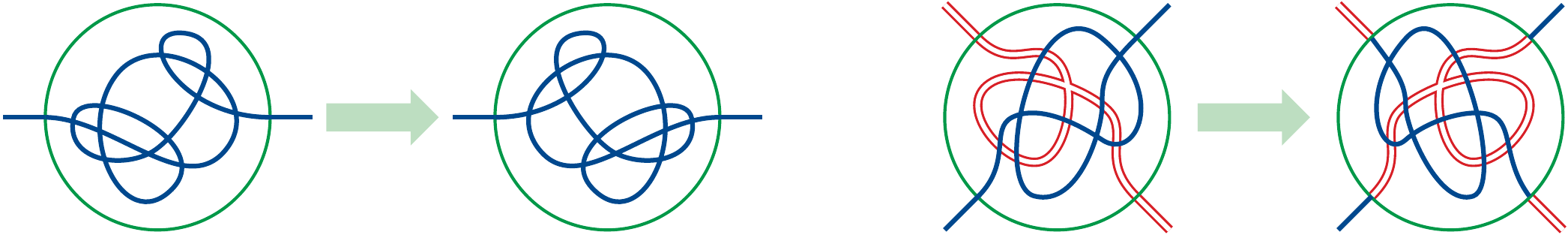}
\caption{Flipping tangles with one and two strands.}
\label{F:flip}
\end{figure}

\begin{lemma}
\label{L:flip1}
Let $\gamma$ be an arbitrary closed curve on the sphere.  Flipping any tangle of $\gamma$ with one strand yields another closed curve $\gamma'$ with $\Defect(\gamma') = \Defect(\gamma)$.
\end{lemma}

\begin{proof}
Let $\sigma$ be a simple closed curve that crosses $\gamma$ at exactly two points.  These points decompose $\sigma$ into two subpaths $\alpha\cdot\beta$, where $\alpha$ is the unique strand of the interior tangle and $\beta$ is the unique strand of the exterior tangle.
Let $\Sigma$ denote the interior disk of $\sigma$, and let $\phi\colon \Sigma\to\Sigma$ denote the homeomorphism that flips the interior tangle.  Flipping the interior tangle yields the closed curve  $\gamma' \coloneqq \emph{rev}(\phi(\alpha))\cdot \beta$, where $\emph{rev}$ denotes path reversal.

No vertex of $\alpha$ is interleaved with a vertex of $\beta$; thus, two vertices in $\gamma'$ are interleaved if and only if the corresponding vertices in $\gamma$ are interleaved.
Every vertex of $\emph{rev}(\phi(\alpha))$ has the same sign as the corresponding vertex of $\alpha$, since both the orientation of the vertex and the order of traversals through the vertex changed.  Thus, every vertex of $\gamma'$ has the same sign as the corresponding vertex of $\gamma$.  We conclude that $\Defect(\gamma') = \Defect(\gamma)$.
\end{proof}

A tangle is \EMPH{tight} if each strand is simple and each pair of strands crosses at most once.  Any tangle can be \EMPH{tightened}---that is, transformed into a tight tangle---by continuously deforming the strands without crossing $\sigma$ or moving their endpoints, and therefore by a finite sequence of homotopy moves.
Let \EMPH{$\gamma\TightenIn\sigma$} and \EMPH{$\gamma\TightenOut\sigma$} denote the closed curves that result from tightening the interior and exterior tangles of $\sigma$, respectively.%
\footnote{We recommend pronouncing $\TightenIn$ as “tightened inside” and $\TightenOut$ as “tightened outside”; note that the symbols $\Cap$ and $\Cup$ resemble the second letters of “inside” and “outside”.}
The following lemma that flipping any 2-strand tangle does not change its defect follows from our inclusion-exclusion formula for defect~\cite[Lemma~5.4]{defect}; we give a simpler proof here to keep the paper self-contained.

\begin{lemma}
\label{L:flip2}
Let $\gamma$ be an arbitrary closed curve on the sphere.  Flipping any tangle of $\gamma$ with two strands yields another closed curve $\gamma'$ with $\Defect(\gamma') = \Defect(\gamma)$.
\end{lemma}

\begin{proof}
Let $\sigma$ be a simple closed curve that crosses $\gamma$ at exactly four points.  These four points naturally decompose $\gamma$ into four subpaths $\alpha \cdot \delta \cdot \beta \cdot \e$, where $\alpha$ and $\beta$ are the strands of the interior tangle of~$\sigma$, and $\delta$ and $\e$ are the strands of the exterior tangle.
Flipping the interior tangle either exchanges $\alpha$ and $\beta$, reverses $\alpha$ and $\beta$, or both; see Figure~\ref{F:flip-cases}.  In every case, the result is a single closed curve $\gamma'$.
We classify each vertex of $\gamma$ as \emph{interior} if it lies on $\alpha$ and/or $\beta$, and \emph{exterior} otherwise.  Similarly, we classify pairs of interleaved vertices are either interior, exterior, or mixed.

\begin{figure}[ht]
\centering
\includegraphics[width=0.75\linewidth]{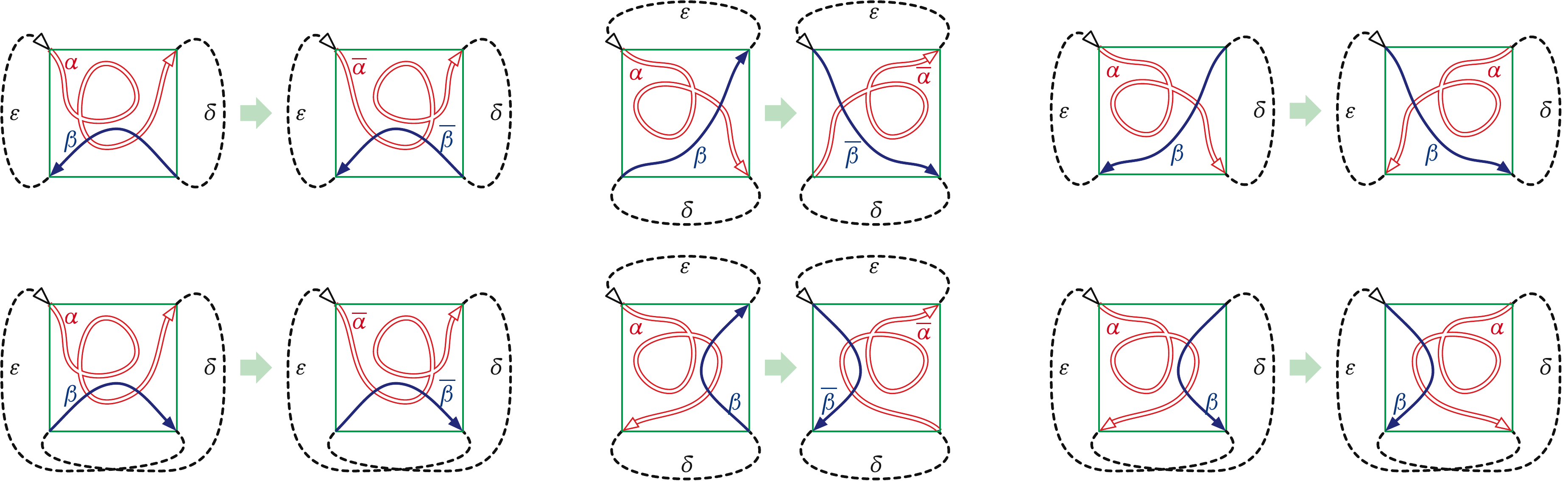}
\caption{Flipping all six types of 2-strand tangle.}
\label{F:flip-cases}
\end{figure}


An interior vertex~$x$ and an exterior vertex $y$ are interleaved if and only if $x$ is an intersection point of $\alpha$ and $\beta$ and $y$ is an intersection point of $\delta$ and $\e$.  Thus, the total contribution of mixed vertex pairs to Polyak's formula $\Defect(\gamma) = -2\sum_{x\between y} \sgn(x)\cdot\sgn(y)$ is
\[
	-2 \sum_{x\in \alpha\cap\beta} \sum_{y\in \delta\cap\e} \sgn(x)\cdot\sgn(y)
	~=~
	-2 \left(\sum_{x\in \alpha\cap\beta}\sgn(x)\right)
			\left(\sum_{y\in \delta\cap\e}\sgn(y)\right).
\]
Consider any sequence of homotopy moves that tightens the interior tangle with strands $\alpha$ and $\beta$.  Any $\arc20$ move involving both $\alpha$ and $\beta$ removes one positive and one negative vertex; no other homotopy move changes the number of vertices in $\alpha\cap\beta$ or the signs of those vertices.  Thus, tightening $\alpha$ and~$\beta$ leaves the sum $\sum_{x\in \alpha\cap \beta} \sgn(x)$  unchanged.
Similarly, tightening the exterior tangle $\delta\cup\e$ leaves the sum $\sum_{y\in \delta\cap\e} \sgn(y)$ unchanged.  But after tightening both tangles, either $\alpha$ and $\beta$ are disjoint, or $\delta$ and $\e$ are disjoint, or both, as $\gamma$ is a single closed curve.  Thus, at least one of the sums $\sum_{x\in \alpha\cap\beta}\sgn(x)$ and $\sum_{y\in \delta\cap\e}\sgn(y)$ is equal to zero.  We conclude that mixed vertex pairs do not contribute to the defect.

The curve $\gamma\TightenIn\sigma$ obtained by tightening $\alpha$ and $\beta$ has at most one interior vertex (and therefore no interior vertex pairs); the exterior vertices of $\gamma\TightenIn\sigma$ are precisely the exterior vertices of $\gamma$.  Similarly, the curve $\gamma\TightenOut\sigma$ obtained by tightening both $\delta$ and $\e$ has at most one exterior vertex; the interior vertices of $\gamma\TightenOut\sigma$ are precisely the interior vertices of $\gamma$.  It follows that
$\Defect(\gamma) = \Defect(\gamma\TightenOut\sigma) + \Defect(\gamma\TightenIn\sigma)$.

Finally, let $\gamma'$ be the result of flipping the interior tangle.  The curve $\gamma'\TightenOut\sigma$ is just a reflection of $\gamma\TightenOut\sigma$, which implies that $\Defect(\gamma'\TightenOut\sigma) = \Defect(\gamma\TightenOut\sigma)$, and straightforward case analysis implies $\gamma'\TightenIn\sigma = \gamma\TightenIn\sigma$.  We conclude that
$\Defect(\gamma') = \Defect(\gamma'\TightenIn\sigma) + \Defect(\gamma'\TightenOut\sigma) = \Defect(\gamma\TightenIn\sigma) + \Defect(\gamma\TightenOut\sigma) = \Defect(\gamma)$.
\end{proof}

Lemmas~\ref{L:navigate}, \ref{L:flip1}, and \ref{L:flip2} now immediately imply Theorem~\ref{Th:flip-defect}.

\subsection{Back to planar electrical moves}

Each planar electrical transformation in a plane graph $G$ induces the same change in the medial graph~$G^\times$ as a finite sequence of 1- and 2-strand tangle flips (hereafter simply called “tangle flips”) followed by a single electrical move.
For an arbitrary connected multicurve $\gamma$, let \EMPH{$\bar{X}(\gamma)$} denote the minimum number of electrical moves in a mixed sequence of electrical moves and tangle flips that tightens $\gamma$.  Similarly, let \EMPH{$\bar{H}(\gamma)$} denote the minimum number of homotopy moves in a mixed sequence of homotopy moves and tangle flips that tightens $\gamma$.
We emphasize that tangle flips are “free” and do not contribute to either $\bar{X}(\gamma)$ or $\bar{H}(\gamma)$.

Our lower bound on planar electrical moves follows our earlier lower bound proof for facial electrical moves almost verbatim; the only subtlety is that the embedding of the graph can effectively change at every step of the reduction.
We repeat the arguments here to keep the presentation self-contained.

\begin{lemma}
\label{L:flip-smoothing}
$\bar{X}(\check\gamma) \le \bar{X}(\gamma)$ for every connected proper smoothing $\check\gamma$ of every connected multicurve $\gamma$ on the sphere.
\end{lemma}

\begin{proof}
Let $\gamma$ be a connected multicurve, and let $\check\gamma$ be a connected proper smoothing of $\gamma$.  The proof proceeds by induction on $\bar{X}(\gamma)$.  If $\bar{X}(\gamma)=0$, then $\gamma$ is already tight, so the lemma is vacuously true.

First, suppose $\check\gamma$ is obtained from $\gamma$ by smoothing a single vertex~$x$.  Consider an optimal mixed sequence of tangle flips and electrical moves that tightens $\gamma$.  This sequence starts with zero or more tangle flips, followed by a electrical move.  Let $\gamma'$ be the multicurve that results from the initial sequence of tangle flips; by definition, we have $\bar{X}(\gamma) = \bar{X}(\gamma')$.
Moreover, applying the same sequence of tangle flips to $\check\gamma$ yields a connected multicurve ${\check\gamma}'$ such that $\bar{X}(\check\gamma) = \bar{X}({\check\gamma}')$.  Thus, we can assume without loss of generality that the first operation in the sequence is an electrical move.

Now let~$\gamma'$ be the result of this move; by definition, we have $\bar{X}(\gamma) = \bar{X}(\gamma')+1$.  As in the proof of Lemma \ref{L:smoothing-annulus}, there are several subcases to consider, depending on whether the move from $\gamma$ to $\gamma'$ involves the smoothed vertex~$x$, and if so, the specific type of move.
In every subcase, by Lemma~\ref{L:smoothing-case} we can apply at most one electrical move to $\check\gamma$ to obtain a (possibly trivial) smoothing ${\check\gamma}'$ of $\gamma'$, and then apply the inductive hypothesis on $\gamma'$ and ${\check\gamma}'$ to prove the statement.  We omit the straightforward details.

Finally, if $\check\gamma$ is obtained from $\gamma$ by smoothing more than one vertex, the lemma follows immediately by induction from the previous analysis.
\end{proof}

\begin{lemma}
\label{L:flip-monotonicity}
For every connected multicurve $\gamma$, there is an intermixed sequence of electrical moves and tangle flips that tightens $\gamma$ that contains exactly $\bar{X}(\gamma)$ electrical moves, and does not contain $\arc01$ or $\arc12$ moves.
\end{lemma}

\begin{proof}
Consider an optimal sequence of electrical moves and tangle flips that tightens~$\gamma$, and let~$\gamma_i$ denote the result of the first $i$ moves in this sequence.  If any $\gamma_i$ has more vertices than its predecessor $\gamma_{i-1}$, then $\gamma_{i-1}$ is a connected proper smoothing of $\gamma_i$, and Lemma~\ref{L:flip-smoothing} implies a contradiction.
\end{proof}

\begin{lemma}
\label{L:flip-XH}
$\bar{X}(\gamma) \ge \bar{H}(\gamma)$ for every closed curve $\gamma$ on the sphere.
\end{lemma}

\begin{proof}
Let $\gamma$ be a closed curve on the sphere.  The proof proceeds by induction on $\bar{X}(\gamma)$.  If $\bar{X}(\gamma) = 0$, then $\gamma$ is simple and thus $\bar{H}(\gamma) = 0$, so assume otherwise.

Consider an optimal sequence of electrical moves and tangle flips that tightens $\gamma$, and let $\gamma_i$ be the curve obtained by applying a prefix of the sequence up to and including the first electrical move.  The minimality of the sequence implies that $\bar{X}(\gamma) = \bar{X}(\gamma')+1$.
By Lemma~\ref{L:flip-monotonicity}, we can assume without loss of generality that the first electrical move in the sequence is neither $\arc01$ nor $\arc12$, and if this first electrical move is $1\arcto 0$ or $3\arcto 3$, the theorem immediately follows by induction.

The only remaining move to consider is $\arc21$.  Let $\gamma^\circ$ denote the result of applying the same sequence of tangle flips to $\gamma$, but replacing the final $\arc21$ move with a $\arc20$ move, or equivalently, smoothing the vertex of $\gamma'$ left by the final $\arc21$ move.
We immediately have $\bar{H}(\gamma) \le \bar{H}(\gamma^\circ) + 1$.  Because $\gamma^\circ$ is a connected proper smoothing of~$\gamma'$, Lemma~\ref{L:flip-smoothing} implies $\bar{X}(\gamma^\circ) < \bar{X}(\gamma') = \bar{X}(\gamma)-1$.  Finally, the inductive hypothesis implies that $\bar{X}(\gamma^\circ) \ge \bar{H}(\gamma^\circ)$, which completes the proof.
\end{proof}

\begin{lemma}
\label{L:flip-Hdefect}
$\bar{H}(\gamma) \ge \abs{\Defect(\gamma)}/2$ for every closed curve $\gamma$ on the sphere.
\end{lemma}

\begin{proof}
Each homotopy move decreases $\abs{\Defect(\gamma)}$ by at most $2$, and Lemmas~\ref{L:flip1} and~\ref{L:flip2} imply that tangle flips do not change $\abs{\Defect(\gamma)}$ at all.  Every simple curve has defect $0$.
\end{proof}


\begin{theorem}
\label{Th:lowerbound-planar}
Let $G$ be an arbitrary planar graph, and let $\gamma$ be any unicursal smoothing of~$G^\times$ (defined with respect to any planar embedding of $G$).  Reducing~$G$ to a single vertex requires at least $\abs{\Defect(\gamma)}/2$ planar electrical transformations.
\end{theorem}

\begin{proof}
The minimum number of planar electrical transformations required to reduce $G$ is at least $\bar{X}(G^\times)$.
Because~$\gamma$ is a single curve, it must be connected, so Lemma~\ref{L:flip-smoothing} implies that $\bar{X}(G^\times) \ge \bar{X}(\gamma)$.  The theorem now follows immediately from Lemmas \ref{L:flip-XH} and \ref{L:flip-Hdefect}.
\end{proof}

Finally, Hayashi \etal~\cite{hhsy-musrm-12} and Even-Zohar \etal~\cite{ehln-irkl-16} describe infinite families of planar closed curves with defect $\Omega(n^{3/2})$; see also \cite[Section 2.2]{tangle}.

\begin{corollary}
\label{C:lowerbound}
Reducing any $n$-vertex planar graph to a single vertex requires $\Omega(n^{3/2})$ planar electrical transformations in the worst case.
\end{corollary}

\section{Open problems}
\label{S:open-problems}

Our results suggest several open problems.  Perhaps the most compelling, and the primary motivation for our work, is to find either a subquadratic upper bound or a quadratic lower bound on the number of (unrestricted) electrical transformations required to reduce any planar graph without terminals to a single vertex.
Like Gitler~\cite{g-dtaa-91}, Feo and Provan~\cite{fp-dtert-93}, and Archdeacon \etal~\cite{acgp-frpwg-00}, we conjecture that $O(n^{3/2})$ \emph{facial} electrical transformations suffice.
However, proving the conjecture appears to be challenging.

Another direction is to prove a quadratic lower bound for graphs on surfaces with positive genus under \emph{crossing-free} electrical transformations.
To generalize Theorem~\ref{Th:flip-defect} to surface-embedded graphs, we need an extension of Lemma~\ref{L:navigate} to navigate through all the possible embeddings.  Using the theory of \emph{large-edgewidth (LEW) embeddings}, a result by Thomassen~{\cite[Theorem~6.1]{t-egnsn-90}} shows that any embedding of a surface-embedded graph can be obtained from the LEW-embedding (if there's one) by a finite sequence of split reflections.
From here it is not hard to construct a toroidal curve that admits an LEW-embedding and has quadratic defect.  The main difficulty is that we don't have a similar electrical-homotopy inequality for arbitrary surfaces.

Finally, none of our lower bound techniques imply anything about non-planar electrical transformations or about electrical reduction of non-planar graphs.  Indeed, the only lower bound known in the most general setting, for \emph{any} family of electrically reducible graphs, is the trivial $\Omega(n)$.
It seems unlikely that planar graphs can be reduced more quickly by using non-planar electrical transformations, but we can't prove anything.  Any non-trivial lower bound for this problem would be interesting.

\bibliographystyle{newuser}

\begin{thebibliography}{10}

\bibitem{a-ccccc-30}
Virgil~W. Adkisson.
\newblock Cyclicly connected continuous curves whose complementary domain
  boundaries are homeomorphic, preserving branch points.
\newblock \emph{C. R. Séances Soc. Sci. Lett. Varsovie III} 23:164--193, 1930.

\bibitem{a-tc-94}
Francesca Aicardi.
\newblock Tree-like curves.
\newblock \emph{Singularities and Bifurcations}, 1--31, 1994. Advances in
  Soviet Mathematics~21, Amer. Math. Soc.

\bibitem{acgp-frpwg-00}
Dan Archdeacon, Charles~J. Colbourn, Isidoro Gitler, and J.~Scott Provan.
\newblock Four-terminal reducibility and projective-planar
  wye-delta-wye-reducible graphs.
\newblock \emph{J. Graph Theory} 33(2):83--93, 2000.

\bibitem{a-pctip-94}
Vladimir~I. Arnold.
\newblock Plane curves, their invariants, perestroikas and classifications.
\newblock \emph{Singularities and Bifurcations}, 33--91, 1994. Adv. Soviet
  Math.~21, Amer. Math. Soc.

\bibitem{a-tipcc-94}
Vladimir~I. Arnold.
\newblock \emph{Topological Invariants of Plane Curves and Caustics}.
\newblock University Lecture Series~5. Amer. Math. Soc., 1994.

\bibitem{c-cepgd-93}
Jaizhen Cai.
\newblock Counting embeddings of planar graphs using {DFS} trees.
\newblock \emph{SIAM J. Discrete Math.} 6(3):335--352, 1993.

\bibitem{cl-chidh-14}
Gregory~R. Chambers and Yevgeny Liokumovich.
\newblock Converting homotopies to isotopies and dividing homotopies in half in
  an effective way.
\newblock \emph{Geometric and Functional Analysis} 24(4):1080--1100. Springer,
  2014.

\bibitem{c-tcgs-18}
Hsien-Chih Chang.
\newblock \emph{Tightening curves and graphs on surfaces}.
\newblock {Ph.D.} dissertation, University of Illinois at Urbana-Champaign,
  2018.

\bibitem{defect}
Hsien-Chih Chang and Jeff Erickson.
\newblock Electrical reduction, homotopy moves, and defect.
\newblock Preprint, October 2015.
\newblock arXiv:\href{http://arxiv.org/abs/1510.00571}{1510.00571}.

\bibitem{tangle}
Hsien-Chih Chang and Jeff Erickson.
\newblock Untangling planar curves.
\newblock \emph{Discrete \& Computational Geometry} 58(4):889--920, 2017.

\bibitem{annulus}
Hsien-Chih Chang and Jeff Erickson.
\newblock Unwinding annular curves and electrically reducing planar networks.
\newblock Accepted to Computational Geometry: Young Researchers Forum, Proc.
  33rd Int. Symp. Comput. Geom., 2017.

\bibitem{untangle}
Hsien-Chih Chang, Jeff Erickson, Arnaud de~Mesmay, David Letscher, Saul
  Schleimer, Eric Sedgwick, Dylan Thurston, and Stephan Tillmann.
\newblock Untangling curves on surfaces via local moves.
\newblock \emph{Proc. 29th Annual ACM-SIAM Symposium on Discrete Algorithms},
  121--135, 2018.

\bibitem{cgv-rep-96}
Yves {Colin de Verdière}, Isidoro Gitler, and Dirk Vertigan.
\newblock Réseaux électriques planaires {II}.
\newblock \emph{Comment. Math. Helvetici} 71:144--167, 1996.

\bibitem{dm-ftpdw-15}
Lino Demasi and Bojan Mohar.
\newblock Four terminal planar {Delta}-{Wye} reducibility via rooted
  {$K_{2,4}$} minors.
\newblock \emph{Proc. 26th Ann. ACM-SIAM Symp. Discrete Algorithms},
  1728--1742, 2015.

\bibitem{dt-ipt-89}
Giuseppe Di~Battista and Roberto Tamassia.
\newblock Incremental planarity testing.
\newblock \emph{Proc. 30th Ann. IEEE Symp. Foundations Comput. Sci.}, 436--441,
  1989.

\bibitem{dt-opt-96}
Giuseppe Di~Battista and Roberto Tamassia.
\newblock On-line planarity testing.
\newblock \emph{SIAM J. Comput.} 25(5):956--997, 1996.

\bibitem{e-rpges-66}
G.~{V}. Epifanov.
\newblock Reduction of a plane graph to an edge by a star-triangle
  transformation.
\newblock \emph{Dokl. Akad. Nauk SSSR} 166:19--22, 1966.
\newblock In Russian. English translation in \textit{Soviet Math. Dokl.}
  7:13--17, 1966.

\bibitem{ehln-irkl-16}
Chaim Even-Zohar, Joel Hass, Nati Linial, and Tahl Nowik.
\newblock Invariants of random knots and links.
\newblock \emph{Discrete \& Computational Geometry} 56(2):274--314, 2016.
\newblock arXiv:\href{http://arxiv.org/abs/1411.3308}{1411.3308}.

\bibitem{f-erpns-85}
Thomas~A. Feo.
\newblock \emph{{I.} {A} {Lagrangian} Relaxation Method for Testing The
  Infeasibility of Certain {VLSI} Routing Problems. {II.} {Efficient} Reduction
  of Planar Networks For Solving Certain Combinatorial Problems}.
\newblock Ph.D. thesis, Univ. California Berkeley, 1985.
\newblock \burl{http://search.proquest.com/docview/303364161}.

\bibitem{fp-dtert-93}
Thomas~A. Feo and J.~Scott Provan.
\newblock Delta-wye transformations and the efficient reduction of two-terminal
  planar graphs.
\newblock \emph{Oper. Res.} 41(3):572--582, 1993.

\bibitem{g-dtaa-91}
Isidoro Gitler.
\newblock \emph{Delta-wye-delta Transformations: Algorithms and Applications}.
\newblock Ph.D. thesis, Department of Combinatorics and Optimization,
  University of Waterloo, 1991.

\bibitem{gs-tdrpg-11}
Isidoro Gitler and Feliú Sagols.
\newblock On terminal delta-wye reducibility of planar graphs.
\newblock \emph{Networks} 57(2):174--186, 2011.

\bibitem{g-gcs-94phdthesis}
Maurits de~Graaf.
\newblock \emph{Graphs and curves on surfaces}.
\newblock {Ph.D.} dissertation, Universiteit van Amsterdam, 1994.

\bibitem{gs-chscc-95}
Maurits de~Graaf and Alexander Schrijver.
\newblock Characterizing homotopy of systems of curves on a compact surface by
  crossing numbers.
\newblock \emph{Linear Alg. Appl.} 226--228:519--528, 1995.

\bibitem{gs-dgs-97}
Maurits de~Graaf and Alexander Schrijver.
\newblock Decomposition of graphs on surfaces.
\newblock \emph{J. Comb. Theory Ser. B} 70:157--165, 1997.

\bibitem{gs-mcmcr-97}
Maurits de~Graaf and Alexander Schrijver.
\newblock Making curves minimally crossing by {Reidemeister} moves.
\newblock \emph{J. Comb. Theory Ser. B} 70(1):134–156, 1997.

\bibitem{hs-ics-85}
Joel Hass and Peter Scott.
\newblock Intersections of curves on surfaces.
\newblock \emph{Israel J. Math.} 51:90--120, 1985.

\bibitem{hs-scs-94}
Joel Hass and Peter Scott.
\newblock Shortening curves on surfaces.
\newblock \emph{Topology} 33(1):25--43, 1994.

\bibitem{hhsy-musrm-12}
Chuichiro Hayashi, Miwa Hayashi, Minori Sawada, and Sayaka Yamada.
\newblock Minimal unknotting sequences of {Reidemeister} moves containing
  unmatched {RII} moves.
\newblock \emph{J. Knot Theory Ramif.} 21(10):1250099 (13 pages), 2012.
\newblock arXiv:\href{http://arxiv.org/abs/1011.3963}{1011.3963}.

\bibitem{h-udtse-35}
Heinz Hopf.
\newblock Über die {Drehung} der {Tangenten} und {Sehnen} ebener {Kurven}.
\newblock \emph{Compositio Math.} 2:50--62, 1935.

\bibitem{k-etscn-1899}
Arthur~Edwin Kennelly.
\newblock Equivalence of triangles and three-pointed stars in conducting
  networks.
\newblock \emph{Electrical World and Engineer} 34(12):413--414, 1899.

\bibitem{m-scpcg-37}
Saunders {Mac Lane}.
\newblock A structural characterization of planar combinatorial graphs.
\newblock \emph{Duke Math. J.} 3(3):460--472, 1937.

\bibitem{nt-aafts-96}
Hiroyuki Nakahara and Hiromitsu Takahashi.
\newblock An algorithm for the solution of a linear system by {$\Delta$-Y}
  transformations.
\newblock \emph{IEICE TRANSACTIONS on Fundamentals of Electronics,
  Communications and Computer Sciences} E79-A(7):1079--1088, 1996.
\newblock Special Section on Multi-dimensional Mobile Information Network.

\bibitem{n-csgs-01}
Max Neumann-Coto.
\newblock A characterization of shortest geodesics on surfaces.
\newblock \emph{Algebraic \& Geometric Topology} 1:349--368, 2001.

\bibitem{nw-kg-00}
Steven~D. Noble and Dominic J.~A. Welsh.
\newblock Knot graphs.
\newblock \emph{J. Graph Theory} 34(1):100--111, 2000.

\bibitem{p-icfgd-98}
Michael Polyak.
\newblock Invariants of curves and fronts via {Gauss} diagrams.
\newblock \emph{Topology} 37(5):989--1009, 1998.

\bibitem{rs-gm7-88}
Neil Robertson and Paul~D. Seymour.
\newblock Graph minors. {VII.} {Disjoint} paths on a surface.
\newblock \emph{J. Comb. Theory Ser. B} 45(2):212--254, 1988.

\bibitem{rv-rse-90}
Neil Robertson and Richard Vitray.
\newblock Representativity of surface embeddings.
\newblock \emph{Paths, Flows, and VLSI-Layout}, 293--328, 1990. Algorithms and
  Combinatorics~9, Springer-Verlag.

\bibitem{s-hcscs-89}
Alexander Schrijver.
\newblock Homotopy and crossing of systems of curves on a surface.
\newblock \emph{Linear Alg. Appl.} 114--115:157--167, 1989.

\bibitem{s-dgshct-91}
Alexander Schrijver.
\newblock Decomposition of graphs on surfaces and a homotopic circulation
  theorem.
\newblock \emph{J. Comb. Theory Ser. B} 51(2):161--210, 1991.

\bibitem{s-cget-92}
Alexander Schrijver.
\newblock Circuits in graphs embedded on the torus.
\newblock \emph{Discrete Math.} 106/107:415--433, 1992.

\bibitem{s-otuok-92}
Alexander Schrijver.
\newblock On the uniqueness of kernels.
\newblock \emph{J. Comb. Theory Ser. B} 55:146--160, 1992.

\bibitem{s-upqtp-85}
Matthias F.~M. Stallmann.
\newblock Using {PQ}-trees for planar embedding problems.
\newblock Tech. Rep. {NCSU-CSC TR-85-24}, Dept. Comput. Sci., NC State Univ.,
  December 1985.
\newblock
  \burl{https://people.engr.ncsu.edu/mfms/Publications/1985-TR_NCSU_CSC-PQ_Trees.pdf}.

\bibitem{s-cpe-93}
Matthias F.~M. Stallmann.
\newblock On counting planar embeddings.
\newblock \emph{Discrete Math.} 122:385--392, 1993.

\bibitem{s-pr-1916}
Ernst Steinitz.
\newblock Polyeder und {Raumeinteilungen}.
\newblock \emph{Enzyklopädie der mathematischen Wissenschaften mit Einschluss
  ihrer Anwendungen} III.AB(12):1--139, 1916.

\bibitem{sr-vtp-34}
Ernst Steinitz and Hans Rademacher.
\newblock \emph{Vorlesungen über die Theorie der Polyeder: unter Einschluß
  der Elemente der Topologie}.
\newblock Grundlehren der mathematischen Wissenschaften~41. Springer-Verlag,
  1934.
\newblock Reprinted 1976.

\bibitem{t-k1-1876}
Peter~Guthrie Tait.
\newblock On knots {I}.
\newblock \emph{Proc. Royal Soc. Edinburgh} 28(1):145--190, 1876–7.

\bibitem{t-egnsn-90}
Carsten Thomassen.
\newblock Embeddings of graphs with no short noncontractible cycles.
\newblock \emph{J. Comb. Theory Ser. B} 48(2):155--177, 1990.

\bibitem{t-w2tg-80}
Klaus Truemper.
\newblock On whitney's 2-isomorphism theorem for graphs.
\newblock \emph{Journal of Graph Theory} 4(1):43--49, 1980.

\bibitem{t-drpg-89}
Klaus Truemper.
\newblock On the delta-wye reduction for planar graphs.
\newblock \emph{J. Graph Theory} 13(2):141--148, 1989.

\bibitem{t-md-92}
Klaus Truemper.
\newblock \emph{Matroid Decomposition}.
\newblock Academic Press, 1992.

\bibitem{w-drag-15}
Donald Wagner.
\newblock Delta-wye reduction of almost-planar graphs.
\newblock \emph{Discrete Appl. Math.} 180:158--167, 2015.

\bibitem{w-cgcg-32}
Hassler Whitney.
\newblock Congruent graphs and the connectivity of graphs.
\newblock \emph{Amer. J. Math.} 54(1):150--168, 1932.

\bibitem{w-2g-1933}
Hassler Whitney.
\newblock 2-isomorphic graphs.
\newblock \emph{American Journal of Mathematics} 55(1):245--254, 1933.

\end{thebibliography}
\def\burl#1{$\langle$\url{#1}$\rangle$}

\newpage
\appendix

\section{Equivalence between electrical and homotopic tightness for primitive curves}
\label{SS:elec-homo-reduced}

A closed curve $\gamma$ is \EMPH{primitive} if $\gamma$ is not homotopic to a proper multiple of some other closed curve.  A multicurve is \emph{primitive} if all its constituent curves are primitive.
We show a five-way equivalence between electrical and homotopic tightness for primitive multicurves, which is implicit in the work by de Graaf and Schrijver \cite{s-hcscs-89,s-dgshct-91,s-otuok-92,s-cget-92,gs-chscc-95,gs-dgs-97,gs-mcmcr-97}.

Let $\gamma$ be a multicurve on an orientable surface $\Sigma$ such that each constituent curve of $\gamma$ is primitive.
Define the \EMPH{$\mu$\,-function} as
\[
\EMPH{$\mu(\gamma, \sigma)$} ~\coloneqq~ \min_{\substack{\sigma' \sim \sigma \\ \sigma' \pitchfork\, \gamma}} \Cr(\gamma, \sigma'),
\]
where $\Cr(\gamma, \sigma')$ is the number of crossings between $\gamma$ and $\sigma'$, and the minimum ranges over every closed curve~$\sigma'$ homotopic to the given closed curve $\sigma$ on $\Sigma$, intersecting $\gamma$ transversely.%
\footnote{In Schrijver \cite{s-otuok-92}, the $\mu$-function is defined with respect to the graph corresponding to $\gamma$ through medial construction; the function defined here is denoted as $\mu'$ in his paper.}
Denote \EMPH{$\mu_\gamma$} as the single-variable function $\mu(\gamma, \cdot)$.
The notion of $\mu$-function is deeply related to the \emph{representativity} or \emph{facewidth} of a graph studied in topological graph theory \cite{rs-gm7-88,rv-rse-90,t-egnsn-90}.



The $\mu$-function is a higher-genus analogue to the $\Depth$ function defined in the annulus (see Section~\ref{SS:tight-annulus}); in particular, both $\mu$ and $\Depth$ are invariant under isotopy of $\gamma$ and the electrical moves \cite{rv-rse-90}.

\begin{lemma}[Robertson and Vitray~{\cite[Proposition~14.4]{rv-rse-90}}]
\label{L:mu}
Electrical moves do not change $\mu_\gamma$ for any multicurve $\gamma$ on surface $\Sigma$.
\end{lemma}

\begin{proof}
For any face of $\gamma$ intersected by some closed curve $\sigma$ that could be deleted after an electrical move, exhaustive case analysis implies that there is another closed curve $\sigma'$ that avoids that face.
\end{proof}

Multicurve $\gamma$ satisfies \EMPH{simplicity conditions} \cite{s-dgshct-91} if
(1)
any lifting of $\gamma_i$ in the universal cover $\hat\Sigma$ does not self-intersect for any constituent curve $\gamma_i$ of $\gamma$, and
(2)
any distinct liftings of $\gamma_i$ and $\gamma_j$ in $\hat\Sigma$ intersect each other at most once for any pair of (possibly identical) constituent curves $\gamma_i$ and $\gamma_j$ of $\gamma$.
%
Multicurve $\gamma$ is \EMPH{minimally crossing} \cite{s-dgshct-91,s-otuok-92} if each constituent curve of $\gamma$ has minimum number of self-intersections in its homotopy class, and every pair of constituent curves has minimum intersections with each other, in their own homotopy classes.  In notation, one has
\[
\Cr(\gamma_i) ~=~ \min_{\gamma'_i \sim \gamma_i} \Cr(\gamma'_i)
\qquad \text{and} \qquad
\Cr(\gamma_i,\gamma_j) ~=~ \min_{\substack{\gamma'_i \sim \gamma_i \\ \gamma'_j \sim \gamma_j}}\Cr(\gamma'_i, \gamma'_j)
\]
for all constituent curves $\gamma_i$ and $\gamma_j$ of $\gamma$; $\Cr(\gamma_i)$ denotes the number of self-intersections of curve $\gamma_i$.
Multicurve $\gamma$ is \EMPH{crossing-tight} \cite{s-dgshct-91,s-otuok-92} if $\mu_\gamma \neq \mu_{\check\gamma}$ for any proper smoothing $\check\gamma$ of $\gamma$.


Our proof of equivalence relies on machineries developed extensively in the sequence of work by de Graaf and Schrijver \cite{s-hcscs-89,s-dgshct-91,s-otuok-92,s-cget-92,gs-chscc-95,gs-dgs-97,gs-mcmcr-97} who did all the weight-lifting.
However the original work does not address the problem of relating electrical and homotopy moves.

\begin{theorem}
\label{Th:tight-equiv}
Let $\gamma$ be a multicurve on an orientable surface whose constituent curves are all primitive.
The following statements are equivalent: (1) Multicurve $\gamma$ satisfies simplicity conditions, (2) $\gamma$ is minimally crossing, (3) $\gamma$ is crossing-tight, (4) $\gamma$ is e-tight, and (5) $\gamma$ is h-tight.
\end{theorem}

\begin{proof}
{$(1) \Leftrightarrow (2) \Leftrightarrow (3)$:}
Schrijver \cite[Proposition~12]{s-dgshct-91} showed that $\gamma$ satisfies simplicity conditions if and only if $\gamma$ is minimally crossing and each constituent curve is primitive.  Later in the same paper \cite[Theorem~5]{s-dgshct-91} he also showed that $\gamma$ is minimally crossing and each constituent curve is primitive if and only if $\gamma$ is crossing-tight.
An alternative proof using the monotonicity of homotopy process can be found in de Graaf's thesis \cite{g-gcs-94phdthesis}.

{$(3) \Rightarrow (4)$:}
In another paper Schrijver \cite[Theorem~2]{s-otuok-92} showed that two crossing-tight multicurves $\gamma$ and~$\gamma'$ can be transformed into each other using only $\arc33$ moves if (and only if) $\mu_\gamma = \mu_{\check\gamma}$.  This result implies that if multicurve $\gamma$ is crossing-tight then $\gamma$ is e-tight, as electrical moves preserves the $\mu$-function by Lemma~\ref{L:mu}.

{$(4) \Rightarrow (5)$:}
Any e-tight multicurve must be h-tight by de Graaf and Schrijver \cite{gs-mcmcr-97} (see Lemma~\ref{L:electric-homotopy}).

{$(5) \Rightarrow (1)$:}
If $\gamma$ is h-tight and primitive, then by Hass and Scott \cite[Lemma~3.4]{hs-ics-85} multicurve $\gamma$ satisfies simplicity conditions.
To elaborate, assume for contradiction that $\gamma$ violates the simplicity conditions.  As $\gamma$ is h-tight one can push each constituent curve of $\gamma$ close to its unique geodesic on the surface without even decreases the number of vertices, similar to the algorithm of de Graaf and Schrijver \cite{gs-mcmcr-97}.
Therefore all the intersections between lifts of constituent curves of $\gamma$ remains after the push.  The primitiveness of the curve $\gamma$ guarantees that each lift of any constituent curve does not self-intersect, and two different lifts of the same constituent curve intersects at most once on $\hat\Sigma$.
Between the lifts of two distinct geodesics there is at most one intersection in the universal cover, and thus the same holds for the lifts of two distinct constituent curves of $\gamma$.
%
%
This concludes the proof.
\end{proof}

\section{Proving Lemma~\ref{L:smoothing-case}}
\label{S:smoothing-proof}


\begin{proof}
We prove the statement by induction on the number of electrical moves in the sequence and the number of smoothed vertices.  If $\check\gamma = \gamma$ then the statement trivially holds.
Otherwise, we first consider the special case where $\check\gamma$ is obtained from $\gamma$ by smoothing a single vertex~$x$.  Without loss of generality let $\gamma'$ be the result of the first electrical move.  There are two nontrivial cases to consider.

First, suppose the move from $\gamma$ to $\gamma'$ does not involve the smoothed vertex $x$.  Then we can apply the same move to $\check\gamma$ to obtain a new multicurve ${\check\gamma}'$; the same multicurve can also be obtained from $\gamma'$ by smoothing~$x$.

\begin{figure}[ht]
\centering
\includegraphics[width=0.8\textwidth]{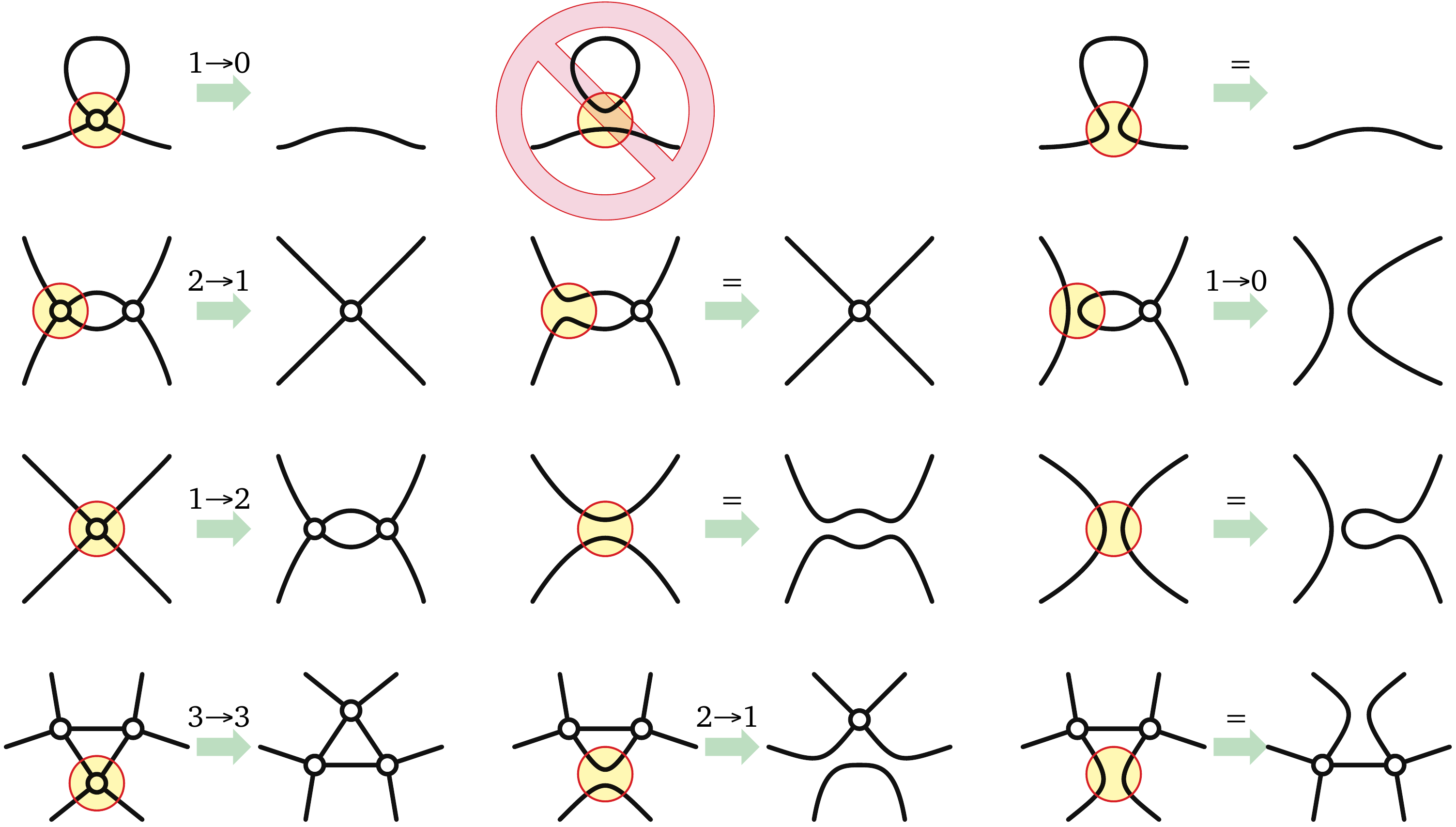}
\caption{Cases for the proof of the Lemma~\ref{L:smoothing-case};
the circled vertex is $x$.}
\label{F:smooth-moves}
\end{figure}

Now suppose the first move does involve $x$.  In this case, we can apply at most one electrical move to $\check\gamma$ to obtain a (possibly trivial) smoothing ${\check\gamma}'$ of $\gamma'$.
There are eight subcases to consider, shown in Figure \ref{F:smooth-moves}.
One subcase for the $\arc10$ move is impossible, because~$\check\gamma$ is connected.
In the remaining $\arc10$ subcase and one $\arc21$ subcase, the curves $\check\gamma$, ${\check\gamma}'$, and $\gamma'$ are all isomorphic.
In all remaining subcases, ${\check\gamma}'$ is a connected proper smoothing of $\gamma'$.

Finally, we consider the more general case where $\check\gamma$ is obtained from $\gamma$ by smoothing more than one vertex.  Let~$\tilde{\gamma}$ be any intermediate curve, obtained from~$\gamma$ by smoothing just one of the vertices that were smoothed to obtain~$\check\gamma$.  As $\check\gamma$ is a connected smoothing of $\tilde{\gamma}$, the curve $\tilde{\gamma}$ itself must be connected too.
Our earlier argument implies that there is a sequence of electrical moves that changes
$\tilde\gamma$ to a smoothing $\tilde\gamma'$ of $\gamma'$.
The inductive hypothesis implies that there is a sequence of electrical moves that changes
$\check\gamma$ to a smoothing $\check\gamma'$ of $\tilde\gamma'$, which is itself a smoothing of~$\gamma'$.  This completes the proof.
\end{proof}

\end{document}